\documentclass[article,11pt]{article}

\usepackage{float,graphicx,verbatim,fullpage,hyperref,amssymb,amsmath,amsthm,enumerate,multicol, xspace,xcolor,mathtools,soul}
\usepackage[margin=1in]{geometry}
\usepackage{cite,enumitem}
\usepackage{subcaption,thm-restate}
\newcommand*{\email}[1]{\texttt{#1}}

\def\final{0}  
\def\iflong{\iffalse}
\ifnum\final=0  
\newcommand{\vnote}[1]{{\color{red}[{\tiny Vivek: \bf #1}]\marginpar{\color{red}*}}}
\newcommand{\knote}[1]{{\color{red}[{\tiny Karthik: \bf #1}]\marginpar{\color{red}*}}}
\newcommand{\krnote}[1]{{\color{red}[{\tiny Krist\'{o}f: \bf #1}]\marginpar{\color{red}*}}}
\newcommand{\tnote}[1]{{\color{red}[{\tiny Tam\'{a}s: \bf #1}]\marginpar{\color{red}*}}}
\newcommand{\todo}[1]{{\color{red}[{\tiny TODO: \bf #1}]\marginpar{\color{red}*}}}
\else 
\newcommand{\vnote}[1]{}
\newcommand{\knote}[1]{}
\newcommand{\krnote}[1]{}
\newcommand{\tnote}[1]{}
\newcommand{\enote}[1]{}
\newcommand{\todo}[1]{}
\fi

\newcommand{\para}[1]{\medskip \noindent {\bf #1}}

\newtheorem{theorem}{Theorem}[section]
\newtheorem{prop}[theorem]{Proposition}
\newtheorem{lemma}[theorem]{Lemma}
\newtheorem{claim}[theorem]{Claim}
\newtheorem{corollary}[theorem]{Corollary}

\theoremstyle{definition}
\newtheorem{definition}[theorem]{Definition}

\def\R{\mathbb{R}}

\def\Z{\mathbb{Z}}
\def\P{\mathcal{P}}
\def\prob{\text{Pr}}
\def\supp{\text{Support}}
\def\face{\text{Face}}
\def\closure{\text{Closure}}
\def\ourgap{1.20016}
\def\ctwo{0.4}
\def\c{c}
\def\VFace{\Delta_{3,n}}
\def\EFace{E_{3,n}}
\def\G{\mathcal{G}'}

\def\alphabound{0.66}
\def\cornercut{fragmenting cut\xspace}
\def\noncornercut{non-fragmenting cut\xspace}
\def\errorterm{\frac{1}{n}}

\DeclarePairedDelimiter\ceil{\lceil}{\rceil}
\DeclarePairedDelimiter\floor{\lfloor}{\rfloor}

\begin{document}
\title{Improving the Integrality Gap for Multiway Cut}
\author{
Krist\'{o}f B\'{e}rczi\thanks{MTA-ELTE Egerv\'{a}ry Research Group, Department of Operations Research, E\"otv\"os Lor\'and University, Budapest, Email: \email{\{berkri,tkiraly\}@cs.elte.hu}.}
\and
Karthekeyan Chandrasekaran\thanks{University of Illinois, Urbana-Champaign,  Email: \email{\{karthe,vmadan2\}@illinois.edu}. Vivek is supported by the NSF grant CCF-1319376.}
\and
Tam\'{a}s Kir\'{a}ly\footnotemark[1]
\and
Vivek Madan\footnotemark[2]
}
\date{}
\maketitle

\begin{abstract}
In the multiway cut problem, we are given an undirected graph with non-negative edge weights and a collection of $k$ terminal nodes, and the goal is to partition the node set of the graph into $k$ non-empty parts each containing exactly one terminal so that the total weight of the edges crossing the partition is minimized. The multiway cut problem for $k\ge 3$ is APX-hard.
For arbitrary $k$, the best-known approximation factor is $1.2965$ due to Sharma and Vondr\'{a}k \cite{SV14} while the best known inapproximability factor is $1.2$ due to Angelidakis, Makarychev and Manurangsi \cite{AMM17}. In this work, we improve on the lower bound to $\ourgap$ by constructing an integrality gap instance for the CKR relaxation.

A technical challenge in improving the gap has been the lack of geometric tools to understand higher-dimensional simplices.
Our instance is a non-trivial $3$-dimensional instance that overcomes this technical challenge. We analyze the gap of the instance by viewing it as a convex combination of $2$-dimensional instances and a uniform 3-dimensional instance.
We believe that this technique could be exploited further to construct instances with larger integrality gap.
One of the 
ingredients of our proof technique 
is a generalization of
a result on \emph{Sperner admissible labelings} due to Mirzakhani and Vondr\'{a}k \cite{MV15} that might be of independent combinatorial interest.
\end{abstract} 
\section{Introduction}
In the multiway cut problem, we are given an undirected graph with non-negative edge weights and a collection of $k$ terminal nodes and the goal is to find a minimum weight subset of edges to delete so that the $k$ input terminals cannot reach each other. For convenience, we will use \emph{$k$-way cut} to denote this problem when we would like to highlight the dependence on $k$ and
\emph{multiway cut} to denote this problem when $k$ grows with the size of the input graph.
The $2$-way cut problem is the classic minimum $\{s,t\}$-cut problem which is solvable in polynomial time.
For $k\ge 3$, Dahlhaus, Johnson, Papadimitriou, Seymour and Yannakakis \cite{DJPSY94} showed that the $k$-way cut problem is APX-hard
and gave a $(2-2/k)$-approximation.
Owing to its applications in partitioning and clustering, $k$-way cut has been an intensely investigated problem in the algorithms literature.
Several novel rounding techniques in the approximation literature were discovered to address the approximability of this problem.


The known approximability as well as inapproximability results are based on a linear programming relaxation, popularly known as the CKR relaxation in honor of the authors---C{\u{a}}linescu, Karloff and Rabani---who introduced it \cite{CKR00}. The CKR relaxation takes a geometric perspective of the problem. For a graph $G=(V,E)$ with edge weights $w:E\rightarrow \R_+$ and terminals $t_1,\ldots, t_k$, the CKR relaxation is given by
\begin{align*}
\min\ &\frac{1}{2}\sum_{e=\{u,v\}\in E}w(e)\|x^u-x^v\|_1\\
x^u&\in \Delta_k\ \forall\ u\in V,\\
x^{t_i}&= e^i\ \forall\ i\in [k],
\end{align*}
where $\Delta_k\coloneqq\{(x_1,\ldots,x_k)\in [0,1]^k:\ \sum_{i=1}^k x_i=1\}$ is the $(k-1)$-dimensional simplex and $e^i\in \{0,1\}^k$ is the extreme point of the simplex along the $i$-th coordinate axis, i.e., $e^i_j=1$ if and only if $j=i$.

C{\u{a}}linescu, Karloff and Rabani designed a rounding scheme for the relaxation which led to a $(3/2-1/k)$-approximation thus improving on the $(2-2/k)$-approximation by Dahlhaus et al. For $3$-way cut, Cheung, Cunningham and Tang \cite{CCT06} as well as Karger, Klein, Stein, Thorup and Young \cite{KKSTY04} designed alternative rounding schemes that led to a $12/11$-approximation factor and also exhibited matching integrality gap instances. We recall that the integrality gap of an instance to the LP is the ratio between the integral optimum value and the LP optimum value. Determining the exact integrality gap of the CKR relaxation for $k\ge 4$ has been an intriguing open question. After the results by Karger et al. and Cunningham et al., a rich variety of rounding techniques were developed to improve the approximation factor of $k$-way cut for $k\ge 4$ \cite{BNS13,SV14,BSW17}. The known approximation factor for multiway cut is $1.2965$ due to Sharma and Vondr\'{a}k \cite{SV14}.

On the hardness of approximation side, Manokaran, Naor, Raghavendra and Schwartz \cite{MNRS08} showed that the hardness of approximation for $k$-way cut is at least the integrality gap of the CKR relaxation  assuming the Unique Games Conjecture (UGC). More precisely, if the integrality gap of the CKR relaxation for $k$-way cut is $\tau_k$, then it is UGC-hard to approximate $k$-way cut within a factor of $\tau_k-\epsilon$ for every constant $\epsilon>0$. As an immediate consequence of this result, we know that the $12/11$-approximation factor for $3$-way cut is tight. 
For $k$-way cut, Freund and Karloff \cite{FK00} constructed an instance showing an integrality gap of $8/(7+(1/(k-1)))$. This was the best known integrality gap until last year when Angelidakis, Makarychev and Manurangsi \cite{AMM17} gave a remarkably simple construction showing an integrality gap of $6/(5+(1/(k-1)))$ for $k$-way cut. In particular, this gives an integrality gap of $1.2$ for multiway cut.

We note that the known upper and lower bounds on the approximation factor for multiway cut match only up to the first decimal digit and thus the approximability of this problem is far from resolved. Indeed Angelidakis, Makarychev and Manurangsi raise the question of whether the lower bound can be improved.
In this work, we improve on the lower bound by constructing an instance with integrality gap \ourgap.
\begin{theorem}\label{theorem:integrality-gap-lower-bound}
For every
constant $\epsilon>0$, there exists an instance of multiway cut such that the integrality gap of the CKR relaxation for that instance is at least $\ourgap-\epsilon$.
\end{theorem}
The above result in conjunction with the result of Manokaran et al. immediately implies that multiway cut is UGC-hard to approximate within a factor of $\ourgap-\epsilon$ for every constant $\epsilon>0$.

One of the ingredients of our technique underlying the proof of Theorem \ref{theorem:integrality-gap-lower-bound} is a generalization of a result on \emph{Sperner admissible labelings} due to Mirzakhani and Vondr\'{a}k \cite{MV15} that might be of independent combinatorial interest (see Theorem \ref{thm:nonmon}).

\section{Background and Result}
Before outlining our techniques, we briefly summarize the background literature that we build upon to construct our instance. We rely on two significant results from the literature.
In the context of the $k$-way cut problem, a \emph{cut} is a function $P:\Delta_k\rightarrow [k+1]$ such that $P(e^i)=i$ for all $i\in [k]$, where we use the notation $[k]:=\{1,2,\dots,k\}$. 
The use of $k+1$ labels as opposed to $k$ labels to describe a cut is a bit non-standard, but is useful for reasons that will become clear later on. 
The approximation ratio $\tau_k(\P)$ of a distribution $\P$ over cuts is given by its \emph{maximum density}:
\[
\tau_k(\P)\coloneqq \sup_{x,y\in \Delta_k, x\neq y} \frac{\prob_{P\sim \P}(P(x)\neq P(y))}{(1/2)\|x-y\|_1}.
\]
Karger et al. \cite{KKSTY04} define 
\[
\tau_k^* \coloneqq  \inf_{\P} \tau_k(\P),
\]
and moreover showed that there exists $\P$ that achieves the infimum. Hence, $\tau_k^*=\min_{\P}\tau_k(\P)$. 
With this definition of $\tau_k^*$, Karger et al. \cite{KKSTY04} showed that for every $\epsilon>0$, there is an instance of multiway cut with $k$ terminals for which the integrality gap of the CKR relaxation is at least $\tau_k^*-\epsilon$. 
Thus, Karger et al.'s result reduced the problem of constructing an integrality gap instance for multiway cut to proving a lower bound on $\tau_k^*$.

Next, Angelidakis, Makarychev and Manurangsi \cite{AMM17} reduced the problem of lower bounding $\tau_k^*$ further by showing that it is sufficient to restrict our attention to \emph{non-opposite cuts} as opposed to all cuts. A cut $P$ is a \emph{non-opposite cut} if $P(x)\in \supp(x)\cup \{k+1\}$ for every $x\in \Delta_k$. Let $\Delta_{k,n}\coloneqq \Delta_k\cap ((1/n) \Z)^k$.
For a distribution $\P$ over cuts, let
\begin{align*}
\tau_{k,n}(\P)&\coloneqq \max_{x,y\in \Delta_{k,n}, x\neq y} \frac{\prob_{P\sim \P}(P(x)\neq P(y))}{(1/2)\|x-y\|_1},\text{ and }\\
\tilde{\tau}_{k,n}^*&\coloneqq \min\{\tau_{k,n}(\P):\P\text{ is a distribution over non-opposite cuts}\}.
\end{align*}
Angelidakis, Makarychev and Manurangsi showed that $\tilde{\tau}_{k,n}^*-\tau_K^* = O(kn/(K-k))$ for all $K>k$. Thus, in order to lower bound $\tau_K^*$, it suffices to lower bound $\tilde{\tau}_{k,n}^*$. That is, it suffices to construct an instance that has \emph{large integrality gap against non-opposite cuts}.

As a central contribution, Angelidakis, Makarychev and Manurangsi constructed an instance showing that $\tilde{\tau}_{3,n}^*\ge 1.2-O(1/n)$. Now, by setting $n=\Theta(\sqrt{K})$, we see that $\tau_K^*$ is at least $1.2-O(1/\sqrt{K})$.
Furthermore, they also showed that their lower bound on $\tilde{\tau}_{3,n}^*$ is almost tight, i.e., $\tilde{\tau}_{3,n}^*\le 1.2$. The salient feature of this framework
is that in order to improve the lower bound on $\tau_K^*$, it suffices to improve $\tilde{\tau}_{k,n}^*$
 for some $4\leq k < K$.

The main technical challenge towards improving $\tilde{\tau}_{4,n}^*$ is that one has to deal with the $3$-dimensional simplex $\Delta_4$. Indeed, all known gap instances including that of Angelidakis, Makarychev and Manurangsi are constructed using the $2$-dimensional simplex. In the $2$-dimensional simplex, the properties of non-opposite cuts are easy to visualize and their cut-values are convenient to characterize using simple geometric observations. However, the values of non-opposite cuts in the $3$-dimensional simplex become difficult to characterize. Our main contribution is a simple argument based on properties of lower-dimensional simplices that overcomes this technical challenge. We construct a $3$-dimensional instance that has gap larger than 1.2 against non-opposite cuts.

\begin{theorem}\label{thm:4-way-cut-gap-against-non-opposite-cuts}
$\tilde{\tau}_{4,n}^*\ge \ourgap - O(1/n)$.
\end{theorem}

Theorem \ref{theorem:integrality-gap-lower-bound} follows from Theorem \ref{thm:4-way-cut-gap-against-non-opposite-cuts} using the above arguments.

\section{Outline of Ideas}
Let $G=(V,E)$ be the graph with node set $\Delta_{4,n}$ and edge set $E_{4,n}:=\{xy:x,y\in \Delta_{4,n}, \|x-y\|_1=2/n\}$, where the terminals are the four unit vectors. 
In order to lower bound $\tilde{\tau}_{4,n}^*$, we will come up with weights on the edges of $G$ such that every non-opposite cut has cost at least $\alpha=\ourgap$ and moreover the cumulative weight of all edges is $n+O(1)$. 
This suffices to lower bound $\tilde{\tau}_{4,n}^*$ by the following proposition.
\begin{prop}\label{prop:weight-to-gap}
Suppose that there exist weights $w:E_{4,n}\rightarrow \R_{\ge 0}$ on the edges of $G$ such that every non-opposite cut has cost at least $\alpha$ and the cumulative weight of all edges is $n+O(1)$. Then, $\tilde{\tau}_{4,n}^*\ge \alpha-O(1/n)$.
\end{prop}
\begin{proof}
For an arbitrary distribution $\P$ over non-opposite cuts, we have 
\begin{align*}
\tau_{k,n}(\P)& = \max_{x,y\in \Delta_{k,n}, x\neq y} \frac{\prob_{P\sim \P}(P(x)\neq P(y))}{(1/2)\|x-y\|_1}\\
 &\geq \max_{xy\in E_{4,n}} \frac{\prob_{P\sim \P}(P(x)\neq P(y))}{(1/2)\|x-y\|_1}\\
 &= \max_{xy\in E_{4,n}} \frac{\prob_{P\sim \P}(P(x)\neq P(y))}{1/n}\\
 & \geq \sum_{xy\in E_{4,n}} \frac{w(xy)\prob_{P\sim \P}(P(x)\neq P(y))}{(1/n)(\sum_{e \in E_{4,n}} w(e))}\\
 & \geq \frac{\alpha}{1+O(1/n)}=\alpha -O(1/n),
\end{align*}
where the last inequality follows from the hypothesis that every non-opposite cut has cost at least $\alpha$ and the cumulative weight of all edges is $n+O(1)$. 
\end{proof}

We obtain our weighted instance from four instances that have large gap against different types of cuts, and then compute the convex combination of these instances that gives the best gap against all non-opposite cuts.

All of our four instances are defined as edge-weights on the graph $G=(V,E)$. We identify $\Delta_{3,n}$ with the facet of $\Delta_{4,n}$ defined by $x_4=0$. Our first three instances are 2-dimensional instances, i.e.\ only edges induced by $\Delta_{3,n}$ have positive weight. The fourth instance has uniform weight on $E_{4,n}$.

We first explain the motivation behind Instances 1,2, and 4, since these are easy to explain. Let
\begin{align*}
L_{ij}&:=\{xy\in E_{4,n}: \supp(x), \supp(y)\subseteq \{i,j\}\}.
\end{align*}
\begin{itemize}
  \item Instance 1 is simply the instance of Angelidakis, Makarychev and Manurangsi \cite{AMM17} on $\Delta_{3,n}$. It has gap $1.2-\errorterm$ against all non-opposite cuts, since non-opposite cuts in $\Delta_{4,n}$ induce non-opposite cuts on $\Delta_{3,n}$. Additionally, we show in Lemma \ref{lemma:3-way-instance-properties} that the gap is strictly larger than $1.2$ by a constant if the following two conditions hold:
      \begin{itemize}
      \item there exist $i,j \in [3]$ such that $L_{ij}$ contains only one edge whose end-nodes have different labels (a cut with this property is called a \emph{\noncornercut}), and
      \item $\Delta_{3,n}$ has a lot of nodes with label 5.
      \end{itemize}
  \item Instance 2 has uniform weight on $L_{12}$, $L_{13}$ and $L_{23}$, and 0 on all other edges. Here, a cut in which each $L_{ij}$ contains at least two edges whose end-nodes have different labels (a \emph{\cornercut}) has large weight. Consequently, this instance has gap at least $2$ against such cuts.
  \item Instance 4 has uniform weight on all edges in $E_{4,n}$. A beautiful result due to Mirzakhani and Vondr\'{a}k \cite{MV15} implies that non-opposite cuts with no node of label 5 have large weight. Consequently, this instance has gap at least $3/2$ against such cuts. We extend their result in Lemma \ref{lemma:non-opposite-cuts-in-discretezed-simplex} to show that the weight remains large if $\Delta_{3,n}$ has few nodes with label 5.
\end{itemize}

At first glance, the arguments above seem to imply that a convex combination of these three instances already gives a gap strictly larger than 1.2 for all non-opposite cuts. 
However, there exist two non-opposite cuts such that at least one of them has cost at most 1.2 in every convex combination of these three instances (see Section \ref{sec:insufficiency-of-3-instances}). 
One of these two cuts is a \cornercut that has almost zero cost in Instance 4 and the best possible cost, namely 1.2, in Instance 1. 
Instance $3$ is constructed specifically to boost the cost against this non-opposite cut. It has positive uniform weight on 3 equilateral triangles, incident to $e^1$, $e^2$ and $e^3$ on the face $\Delta_{3,n}$. We call the edges of these triangles \emph{red edges}. The side length of these triangles is a parameter, denoted by $c$, that is optimized at the end of the proof. Essentially, we show that if a non-opposite cut has small cost both on Instance 1 and Instance 4 (i.e., weight $1.2$ on Instance 1 and $O(1/n^2)$ weight on Instance 4), then it must contain red edges.


Our lower bound of $\ourgap$ is obtained by optimizing the coefficients of the convex combination and the parameter $c$. By Proposition \ref{prop:weight-to-gap} and the results of Angelidakis, Makarychev and Manurangsi, we obtain that $\tau_{K}^*\ge \ourgap-O(1/\sqrt{K})$, i.e., the integrality gap of the CKR relaxation for $k$-way cut is at least $\ourgap-O(1/\sqrt{k})$. We complement our lower bound of $\ourgap$ by also showing that the best possible gap that can be achieved using convex combinations of our four instances is $1.20067$ (see Section \ref{sec:four-instances-limitations}). 


\section{A $3$-dimensional gap instance against non-opposite cuts}
We will focus on the graph $G=(V,E)$ with the node set $V:=\Delta_{4,n}$ being the discretized $3$-dimensional simplex and the edge set $E_{4,n}:=\{xy:x,y\in \Delta_{4,n}, \|x-y\|_1=2/n\}$. The four terminals $s_1,\ldots, s_4$ will be the four extreme points of the simplex, namely $s_i=e^i$ for $i\in [4]$.
In this context, a cut is a function $P:V\rightarrow [5]$ such that $P(s_i)=i$ for all $i\in [4]$. The \emph{cut-set} corresponding to $P$ is defined as
\[\delta(P):=\{xy \in E_{4,n}: P(x) \neq P(y)\}.\]
For a set $S$ of nodes, we will also use $\delta(S)$ to denote the set of edges with exactly one end node in $S$.
Given a weight function $w:E_{4,n} \to \R_+$, the \emph{cost} of a cut $P$ is $\sum_{e \in \delta(P)}w(e)$.
Our goal is to come up with weights on the edges so that the resulting $4$-way cut instance has gap at least $\ourgap$ against non-opposite cuts.

We recall that $L_{ij}$ denotes the boundary edges between terminals $s_i$ and $s_j$, i.e.,
\[
L_{ij}=\{xy\in E_{4,n}: \supp(x), \supp(y)\subseteq \{i,j\}\}.
\]
We will denote the boundary nodes between terminals $s_i$ and $s_j$ as $V_{ij}$, i.e.,
\[
V_{ij}:=\left\{x\in \Delta_{4,n}:\supp(x)\subseteq\{i,j\}\right\}.
\]
Let $\c\in (0,1/2)$ be a constant to be fixed later, such that $cn$ is integral.
For each $k\in [3]$, we define node sets $U_k, R_k$ and $\closure(R_k)$ and edge set $\Gamma_k$ as follows:
\begin{align*}
U_k&:=\{x\in \Delta_{4,n}: x_4=0,\ x_k=1-\c\}, \\
R_k &:= U_k\cup \{x\in V_{ik}\cup V_{jk}: x_k\ge 1-\c\}, \\
\closure(R_k)&:=\{x\in \Delta_{4,n}: x_4=0,\  x_k\ge 1-\c\}, \text{ and}\\
\Gamma_k &:= \left\{xy\in E_{4,n}: x,y\in R_k\right\}.
\end{align*}
We will refer to the nodes in $R_k$ as red\footnote{We use the term ``red'' as a convenient way for the reader to remember these nodes and edges. The exact color is irrelevant. } nodes near terminal $s_k$ and the edges in $\Gamma_k$ as the red edges near terminal $s_k$ (see Figure \ref{fig:red}).
Let $\face(s_1,s_2,s_3)$ denote the subgraph of $G$ induced by the nodes whose support is contained in $\{1,2,3\}$.
We emphasize that red edges and red nodes are present only in $\face(s_1,s_2,s_3)$ and that the total number of red edges is exactly $9cn$.



\begin{figure}
\centering
\begin{subfigure}[t]{.45\textwidth}
  \centering
  \includegraphics[width=.7\linewidth]{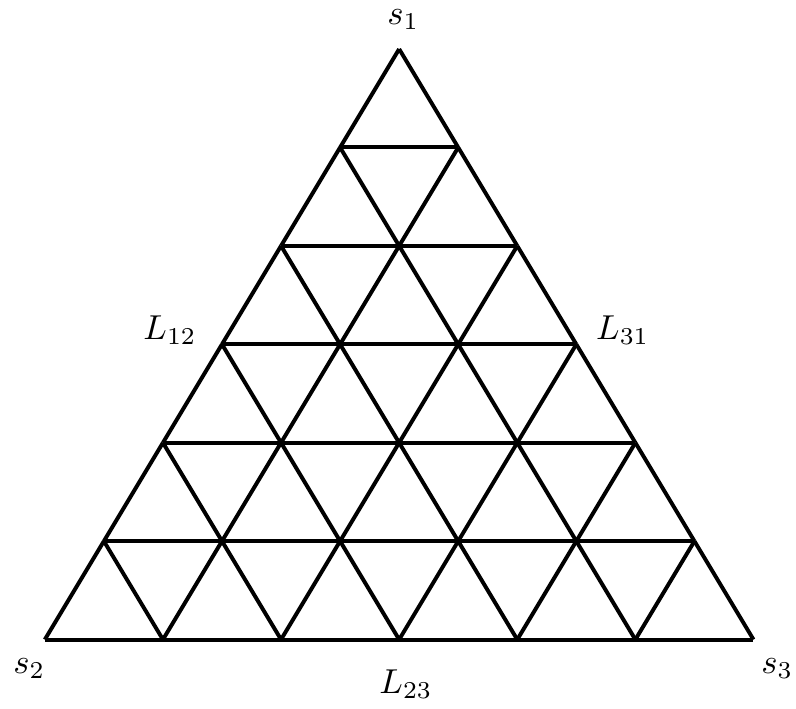}
  \caption{One face of the simplex with edge-sets $L_{12}$, $L_{23}$ and $L_{31}$.}
  \label{fig:triangle}
\end{subfigure}\hfill
\begin{subfigure}[t]{.45\textwidth}
  \centering
  \includegraphics[width=.75\linewidth]{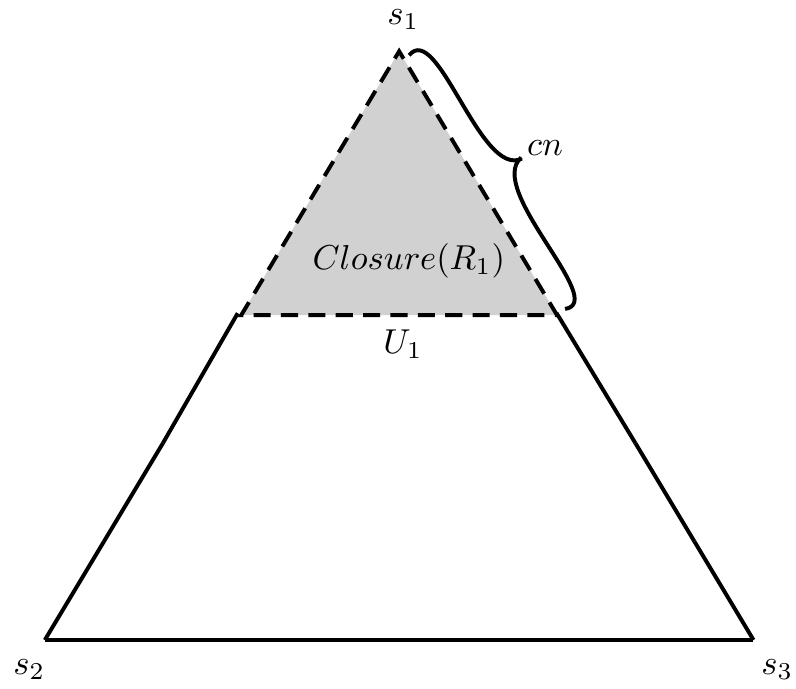}
  \caption{Definition of red nodes and edges near terminal $s_1$. Dashed part corresponds to $(R_1,\Gamma_1)$.}
  \label{fig:red}
\end{subfigure}
\caption{Notation on $\face(s_1,s_2,s_3)$.}
\label{fig:obs}
\end{figure}

\subsection{Gap instance as a convex combination}
Our gap instance is a convex combination of the following four instances.
\begin{enumerate}
\item \textbf{Instance $I_1$.} Our first instance constitutes the $3$-way cut instance constructed by Angelidakis, Makarychev and Manurangsi \cite{AMM17} that has gap $1.2$ against non-opposite cuts. To ensure that the total weight of all the edges in their instance is exactly $n$, we will scale their instance by $6/5$. Let us denote the resulting instance as $J$. In $I_1$, we simply use the instance $J$ on $\face(s_1,s_2,s_3)$ and set the weights of the rest of the edges in $E_{4,n}$ to be zero.

\item \textbf{Instance $I_2$.} In this instance, we set the weights of the edges in $L_{12}, L_{23}, L_{13}$ to be $1/3$ and the weights of the rest of the edges in $E_{4,n}$ to be zero.

\item \textbf{Instance $I_3$.} In this instance, we set the weights of the red edges to be $1/9c$ and the weights of the rest of the edges in $E_{4,n}$ to be zero.

\item \textbf{Instance $I_4$.} In this instance, we set the weight of every edge in $E_{4,n}$ to be $1/n^2$.
\end{enumerate}

We note that the total weight of all edges in each of the above instances is $n+O(1)$.
For multipliers $\lambda_1, \ldots, \lambda_4\ge 0$ to be chosen later that will satisfy $\sum_{i=1}^4 \lambda_i =1$, let the instance $I$ be the convex combination of the above four instances, i.e., $I=\lambda_1 I_1 + \lambda_2 I_2 + \lambda_3 I_3 + \lambda_4 I_4$. By the properties of the four instances, it immediately follows that the total weight of all edges in the instance $I$ is also $n+O(1)$.

\subsection{Gap of the Convex Combination}
The following theorem is the main result of this section.
\begin{theorem}\label{theorem:gap-in-convex-combination}
For every $n\ge 10$ and $c\in (0,1/2)$ such that $cn$ is integer, every non-opposite cut on $I$ has cost at least the minimum of the following two terms:
\begin{enumerate}[label=(\roman*)]
\item $\lambda_2 + (1.2-\errorterm)\lambda_1 + \min_{\alpha\in \left[0,\frac{1}{2}\right]}\left\{\ctwo\alpha \lambda_1+3\left(\frac{1}{2}-\alpha\right)\lambda_4  \right\}$

\item $2\lambda_2 + (1.2-\frac{5}{2n})\lambda_1 + 3\min\left\{\frac{2\lambda_3}{9\c},\min_{\alpha\in \left[0,\frac{\c^2}{2}\right]}\left\{\ctwo\alpha\lambda_1+3\left(\frac{\c^2}{2}-\alpha\right)\lambda_4 \right\} \right\}$
\end{enumerate}
\end{theorem}

Before proving Theorem \ref{theorem:gap-in-convex-combination}, we see its consequence.
\begin{corollary}\label{coro:gap-bound}
There exist constants $\c\in (0,1/2)$ and $\lambda_1, \lambda_2, \lambda_3, \lambda_4\ge 0$ with $\sum_{i=1}^4\lambda_i =1$ such that the cost of every non-opposite cut in the resulting convex combination $I$ is at least $\ourgap -O(1/n)$.
\end{corollary}
\begin{proof}
The corollary follows from Theorem \ref{theorem:gap-in-convex-combination} by setting 
$\lambda_1=0.751652$, 
$\lambda_2=0.147852$, 
$\lambda_3=0.000275$, 
$\lambda_4 =0.100221$ and $c=0.074125$ (this is the optimal setting to achieve the largest lower bound based on Theorem \ref{theorem:gap-in-convex-combination}). 
\end{proof}

Corollary \ref{coro:gap-bound} in conjunction with Proposition \ref{prop:weight-to-gap} immediately implies Theorem \ref{thm:4-way-cut-gap-against-non-opposite-cuts}.

The following theorem (shown in Section \ref{sec:four-instances-limitations}) complements Corollary \ref{coro:gap-bound} by giving an upper bound on the best possible gap that is achievable using the convex combination of our four instances. 

\begin{theorem}\label{thm:limitation}
For every constant $\c\in (0,1/2)$ and every $\lambda_1, \lambda_2, \lambda_3, \lambda_4\ge 0$ with $\sum_{i=1}^4\lambda_i =1$, there exists a non-opposite cut whose cost in the resulting convex combination $I$ is at most $1.20067+O(1/n)$. 
\end{theorem}

In light of Corollary \ref{coro:gap-bound} and Theorem \ref{thm:limitation}, if we believe that the integrality gap of the CKR relaxation is more than $1.20067$, then considering convex combination of alternative instances is a reasonable approach towards proving this. 

\medskip
The rest of the section is devoted to proving Theorem \ref{theorem:gap-in-convex-combination}.
We rely on two main ingredients in the proof.
The first ingredient is a statement about non-opposite cuts in the $3$-dimensional discretized simplex. We prove this in Section \ref{sec:non-opposite-cut-size}, where we also give a generalization to higher dimensional simplices, which might be of independent interest.

\begin{restatable}{lemma}{lemmaNonOppCutSize}
\label{lemma:non-opposite-cuts-in-discretezed-simplex}
Let $P$ be a non-opposite cut on $\Delta_{4,n}$ with $\alpha (n+1)(n+2)$ nodes from $\face(s_1,s_2,s_3)$ labeled as $1$, $2$, or $3$ for some $\alpha\in [0,1/2]$. Then, $|\delta(P)|\ge 3\alpha n (n+1)$.
\end{restatable}

The constant $3$ that appears in the conclusion of Lemma \ref{lemma:non-opposite-cuts-in-discretezed-simplex} is the best possible for any fixed $\alpha$ (if $n \to \infty$). To see this, consider the non-opposite cut $P$ 
obtained by labeling $s_i$ to be $i$ for every $i\in [4]$, all nodes at distance at most $\alpha n$ from $s_1$ to be $1$, and all remaining nodes to be $5$. The number of nodes from $\face(s_1,s_2,s_3)$ labeled as $1$, $2$, or $3$ is $\alpha n^2+O(n)$. The number of edges in the cut is $3\alpha n^2+O(n)$.


The second ingredient involves properties of the $3$-way cut instance constructed by Angelidakis, Makarychev and Manurangsi \cite{AMM17}. We need
two properties that are summarized in Lemma \ref{lemma:3-way-instance-properties} and Corollary \ref{coro:red-island-cut-cost}.
We prove these properties in Section \ref{sec:3-way-instance-properties}.
We define a cut $Q:\Delta_{3,n}\rightarrow [4]$ to be a \emph{\cornercut} if $|\delta(Q)\cap L_{ij}|\ge 2$ for every distinct $i,j\in [3]$; otherwise it is a \noncornercut. We recall that $J$ denotes the instance obtained from the 3-way cut instance of Angelidakis, Makarychev and Manurangsi by scaling it up by $6/5$.



The first property is that non-opposite non-fragmenting cuts in $\Delta_{3,n}$ that label a large number of nodes with label $4$ have cost much larger than $1.2$.
\begin{restatable}{lemma}{lemmaThreeWayProps}
\label{lemma:3-way-instance-properties}
Let $Q:\Delta_{3,n}\rightarrow [4]$ be a non-opposite cut with $\alpha n^2$ nodes labeled as $4$. If $Q$ is a \noncornercut and $n\geq 10$, then the cost of $Q$ on $J$ is at least $1.2+\ctwo\alpha-\errorterm$.
\end{restatable}
We show Lemma \ref{lemma:3-way-instance-properties} by modifying $Q$ to obtain a non-opposite cut $Q'$ while reducing its cost by $0.4\alpha$. By the main result of \cite{AMM17}, the cost of every non-opposite cut $Q'$ on $J$ is at least $1.2-\errorterm$.  
Therefore, it follows that the cost of $Q$ on $J$ is at least $1.2-\errorterm+0.4\alpha$. We emphasize that while it might be possible to improve the constant $0.4$ that appears in the conclusion of Lemma \ref{lemma:3-way-instance-properties}, it does not lead to much improvement on the overall integrality gap as illustrated by the results in Section \ref{sec:four-instances-limitations}. 



The second property is that non-opposite cuts which do not remove any of the red edges, but label a large number of nodes in the red region with label $4$ have cost much larger than $1.2$.
\begin{restatable}{corollary}{coroRedIslandCutCost}
\label{coro:red-island-cut-cost}
Let $Q:\Delta_{3,n}\rightarrow [4]$ be a non-opposite cut and $n\geq 10$. For each $i\in [3]$, let
\begin{equation*}
A_i:=
\begin{cases}
\{v\in \closure(R_i): Q(v)=4\} & \text{ if } \delta(Q)\cap \Gamma_i=\emptyset,\\
\emptyset & \text{ otherwise}.
\end{cases}
\end{equation*}
Then, the cost of $Q$ on $J$ is at least $1.2+\ctwo\sum_{i=1}^3|A_i|/n^2-\frac{5}{2n}$.
\end{restatable}

In order to show Corollary \ref{coro:red-island-cut-cost}, we first derive that the cost of the edges $\delta(\cup_{i=1}^3 A_i)$ in the instance $J$ is at least $0.4\sum_{i=1}^3 |A_i|/n^2-\frac{3}{2n}$ using Lemma \ref{lemma:3-way-instance-properties}. Next, we modify $Q$ to obtain a non-opposite cut $Q'$ such that $\delta(Q')=\delta(Q)\setminus \delta(\cup_{i=1}^3A_i)$. By the main result of \cite{AMM17}, the cost of every non-opposite cut $Q'$ on $J$ is at least $1.2-\errorterm$. Therefore, it follows that the cost of $Q$ on $J$ is at least $1.2+0.4\sum_{i=1}^3 |A_i|/n^2$.

We now have the ingredients to prove Theorem \ref{theorem:gap-in-convex-combination}.
\begin{proof}[Proof of Theorem \ref{theorem:gap-in-convex-combination}]
Let $P:\Delta_{4,n}\rightarrow [5]$ be a non-opposite cut. 
Let $Q$ be the cut $P$ restricted to $\face(s_1,s_2,s_3)$, i.e., for every $v\in \Delta_{4,n}$ with $\supp(v)\subseteq [3]$, let
\[
Q(v):=
\begin{cases}
P(v) & \text{ if } P(v)\in \{1,2,3\},\\
4 & \text{ if } P(v)=5.
\end{cases}
\]
We consider two cases.

\paragraph{Case 1: $Q$ is a \noncornercut.} 
Let the number of nodes in $\face(s_1,s_2,s_3)$ that are labeled by $Q$ as $4$ (equivalently, labeled by $P$ as $5$) be $\alpha (n+1)(n+2)$ for some $\alpha\in \left[0,\frac{1}{2}\right]$.
Since $|\{x \in \face(s_1,s_2,s_3): Q(x)=4\}| \geq \alpha n^2$, Lemma \ref{lemma:3-way-instance-properties} implies that the cost of $Q$ on $J$, and hence the cost of $P$ on $I_1$, is at least $1.2+\ctwo\alpha-\errorterm$. Moreover, the cost of $P$ on $I_2$ is at least $1$ since at least one edge in $L_{ij}$ should be in $\delta(P)$ for every pair of distinct $i,j\in [3]$. To estimate the cost on $I_4$, we observe that the number of nodes on $\face(s_1,s_2,s_3)$ labeled by $P$ as $1$, $2$, or $3$ is $(1/2-\alpha)(n+1)(n+2)$.
By Lemma \ref{lemma:non-opposite-cuts-in-discretezed-simplex}, we have that $|\delta(P)|\ge 3(1/2-\alpha)n(n+1)$ and thus, the cost of $P$ on $I_4$ is at least $3(1/2-\alpha)$. Therefore, the cost of $P$ on the convex combination instance $I$ is at least
\[
\lambda_2 + \left(1.2-\errorterm \right)\lambda_1 + \min_{\alpha\in \left[0,\frac{1}{2}\right]}\left\{ \ctwo\alpha\lambda_1+ 3\left(\frac{1}{2}-\alpha\right)\lambda_4 \right\}.
\]

\paragraph{Case 2: $Q$ is a \cornercut.} Then,
the cost of $P$ on $I_2$ is at least $2$ as a \cornercut contains at least $2$ edges from each $L_{ij}$ for distinct $i,j\in [3]$.

We will now compute the cost of $P$ on the other instances.
Let $r:=|\{i\in [3]: \delta(P)\cap \Gamma_i\neq \emptyset\}|$, i.e., $r$ is the number of red triangles that are intersected by the cut $P$. We will derive lower bounds on the cost of the cut in each of the three instances $I_1, I_3$ and $I_4$ based on the value of $r\in \{0,1,2,3\}$.
For each $i\in [3]$, let
\begin{equation*}
A_i:=
\begin{cases}
\{v\in \closure(R_i): P(v)=5\} & \text{ if } \delta(P)\cap \Gamma_i=\emptyset,\\
\emptyset & \text{ otherwise},
\end{cases}
\end{equation*}
and let $\alpha:=|A_1\cup A_2\cup A_3|/((n+1/\c)(n+2/\c))$. Since $c<1/2$, the sets $A_i$ and $A_j$ are disjoint for distinct $i,j\in [3]$.
We note that $\alpha\in[0,(3-r)\c^2/2]$ since $|A_i|\le (\c n+1)(\c n+2)/2$ and $A_i\cap A_j=\emptyset$.

In order to lower bound the cost of $P$ on $I_1$, we will use Corollary \ref{coro:red-island-cut-cost}. We recall that $Q$ is the cut $P$ restricted to $\face(s_1,s_2,s_3)$, so
the cost of $P$ on $I_1$ is the same as the cost of $Q$ on $J$. Moreover, by Corollary \ref{coro:red-island-cut-cost}, the cost of $Q$ on $J$ is at least $1.2+\ctwo\alpha-\frac{5}{2n}$, because $\alpha \leq \sum_{i=1}^3|A_i|/n^2$. Hence, the cost of $P$ on $I_1$ is at least $1.2+\ctwo\alpha-\frac{5}{2n}$.

The cost of $P$ on $I_3$ is at least $2r/9\c$ by the following claim.
\begin{claim}\label{claim:nonempty-red-two-red}
Let $i\in [3]$. If $\delta(P)\cap \Gamma_i\neq \emptyset$, then $|\delta(P)\cap \Gamma_i|\ge 2$.
\end{claim}
\begin{proof}
The subgraph $(R_i,\Gamma_i)$ is a cycle. If $P(x) \neq P(y)$ for some $xy \in \Gamma_i$, then the path $\Gamma_i-xy$ must also contain two consecutive nodes labeled differently by $P$.
\end{proof}

Next we compute the cost of $P$ on $I_4$.
If $r=3$, then the cost of $P$ on $I_4$ is at least $0$. Suppose $r\in \{0,1,2\}$.
For a red triangle $i\in [3]$ with $\delta(P)\cap \Gamma_i=\emptyset$, we have at least $(\c n+1)(\c n+2)/2-|A_i|$ nodes from $\closure(R_i)$ that are labeled as $1$, $2$, or $3$. Moreover, the nodes in $\closure(R_i)$ and $\closure(R_j)$ are disjoint for distinct $i,j\in[3]$. Hence, the number of nodes in $\face(s_1,s_2,s_3)$ that are labeled as $1$, $2$, or $3$ is at least $(3-r)(\c n+1)(\c n+2)/2-\alpha (n+1/\c)(n+2/\c)=((3-r)\c^2/2-\alpha)(n+1/\c)(n+2/\c)$, which is at least $((3-r)\c^2/2-\alpha)(n+1)(n+2)$, since $\c \leq 1$. Therefore, by Lemma \ref{lemma:non-opposite-cuts-in-discretezed-simplex}, we have $|\delta(P)|\ge 3((3-r)\c^2/2-\alpha)n^2$ and thus, the cost of $P$ on $I_4$ is at least $3((3-r)\c^2/2-\alpha)$.

Thus, the cost of $P$ on the convex combination instance $I$ is at least $2\lambda_2+(1.2-\frac{5}{2n})\lambda_1 + \gamma(r,\alpha)$ for some $\alpha\in [0,(3-r)c^2/2]$, where
\[
\gamma (r,\alpha) :=
\begin{cases}
\frac{6\lambda_3}{9\c}, &\text{ if }r=3, \\
\ctwo\alpha \lambda_1 + \frac{2r}{9\c}\lambda_3 + 3\left(\frac{(3-r)\c^2}{2}-\alpha\right)\lambda_4, &\text{ if }r\in \{0,1,2\}.
\end{cases}
\]
In particular, the cost of $P$ on the convex combination instance $I$ is at least $2\lambda_2+(1.2-5/(2n))\lambda_1 + \gamma^*$, where
\[
\gamma^*:=\min_{r\in\{0,1,2,3\}}\min_{\alpha\in \left[0,\frac{(3-r)\c^2}{2}\right]} \gamma(r,\alpha).
\]
Now, Claim \ref{claim:gamma-star} completes the proof of the theorem.
\end{proof}
\begin{claim}\label{claim:gamma-star}
\[
\gamma^*
\ge 3\min\left\{\frac{2\lambda_3}{9\c},\min_{\alpha\in \left[0,\frac{\c^2}{2}\right]}\left\{\ctwo\alpha\lambda_1+3\left(\frac{\c^2}{2}-\alpha\right)\lambda_4 \right\} \right\}.
\]
\end{claim}
\begin{proof}
Let $\gamma(r):=\min_{\alpha\in [0,(3-r)c^2/2]}\gamma(r,\alpha)$.
If $r=3$, then the claim is clear. We consider the three remaining cases.
\begin{enumerate}[label=(\Roman*)]
\item Say $r=0$. Then,
\begin{align*}
\gamma(0)
&= \min_{\alpha\in \left[0,\frac{3\c^2}{2}\right]}\left\{ \ctwo \alpha \lambda_1 + 3\left(\frac{3\c^2}{2}-\alpha\right)\lambda_4\right\}
= 3\min_{\alpha\in \left[0,\frac{\c^2}{2}\right]}\left\{ \ctwo \alpha \lambda_1 + 3\left(\frac{\c^2}{2}-\alpha\right)\lambda_4\right\}.\\
\end{align*}
\item Say $r=1$. Then,
\begin{align*}
\gamma(1)
&= \min_{\alpha\in \left[0,\c^2\right]}\left\{ \ctwo \alpha \lambda_1 + \frac{2}{9\c}\lambda_3+3\left(\c^2-\alpha\right)\lambda_4\right\}\\
&= \frac{2}{9\c}\lambda_3+\min_{\alpha\in \left[0,\frac{\c^2}{2}\right]}\left\{ 2\cdot\ctwo \alpha \lambda_1 + 3\left(\c^2-2\alpha\right)\lambda_4\right\}\\
&= \frac{2}{9\c}\lambda_3+2\min_{\alpha\in \left[0,\frac{\c^2}{2}\right]}\left\{ \ctwo \alpha \lambda_1 + 3\left(\frac{\c^2}{2}-\alpha\right)\lambda_4\right\}\\
&\ge 3\min\left\{\frac{2\lambda_3}{9\c},\min_{\alpha\in \left[0,\frac{\c^2}{2}\right]}\left\{\ctwo\alpha\lambda_1+3\left(\frac{\c^2}{2}-\alpha\right)\lambda_4 \right\} \right\},
\end{align*}
where the last inequality is from the identity $x+2y\ge 3\min\{x,y\}$ for all $x,y\in \R$.
\item Say $r=2$. Then,
\begin{align*}
\gamma(2)
&= \min_{\alpha\in \left[0,\frac{\c^2}{2}\right]}\left\{ \ctwo \alpha \lambda_1 + \frac{4}{9\c}\lambda_3+3\left(\frac{\c^2}{2}-\alpha\right)\lambda_4\right\}\\
&= \frac{4}{9\c}\lambda_3+\min_{\alpha\in \left[0,\frac{\c^2}{2}\right]}\left\{ \ctwo \alpha \lambda_1 + 3\left(\frac{\c^2}{2}-\alpha\right)\lambda_4\right\}\\
&\ge 3\min\left\{\frac{2\lambda_3}{9\c},\min_{\alpha\in \left[0,\frac{\c^2}{2}\right]}\left\{\ctwo\alpha\lambda_1+3\left(\frac{\c^2}{2}-\alpha\right)\lambda_4 \right\} \right\},
\end{align*}
where the last inequality is from the identity $2x+y\ge 3\min\{x,y\}$ for all $x,y\in \R$.

\end{enumerate}
\end{proof}

\section{Size of non-opposite cuts in $\Delta_{k,n}$} \label{sec:non-opposite-cut-size}
In this section, we prove Lemma \ref{lemma:non-opposite-cuts-in-discretezed-simplex}. In fact, we prove a general result for $\Delta_{k,n}$, that may be useful for obtaining improved bounds by considering higher dimensional simplices. Our result is an extension of a theorem of Mirzakhani and Vondr\'ak \cite{MV15} on Sperner-admissible labelings.

A labeling $\ell: \Delta_{k,n} \to [k]$ is \emph{Sperner-admissible} if $\ell(x) \in \supp(x)$ for every $x \in \Delta_{k,n}$. We say that $x\in \Delta_{k,n}$ has an \emph{inadmissible label} if $\ell(x) \notin \supp(x)$. Let $H_{k,n}$ denote the hypergraph whose node set is $\Delta_{k,n}$ and whose hyperedge set is
\[\mathcal{E}:=\left\{\left\{\frac{n-1}{n}x+\frac{1}{n}e_1,\frac{n-1}{n}x+\frac{1}{n}e_2,\dots,
\frac{n-1}{n}x+\frac{1}{n}e_k\right\}: x \in \Delta_{k,n-1}\right\}.\]
Each hyperedge $e \in \mathcal{E}$ has $k$ nodes, and if $x,y \in e$, then there exist distinct $i,j\in[n]$ such that $x-y=\frac{1}{n}e_i-\frac{1}{n}e_j$.  We remark that $H_{k,n}$ has $\binom{n+k-1}{k-1}$ nodes and $\binom{n+k-2}{k-1}$ hyperedges. Geometrically, the hyperedges correspond to simplices that are translates of each other and share at most one node.
Given a labeling $\ell$, a hyperedge of $H_{k,n}$ is \emph{monochromatic} if all of its nodes have the same label. Mirzakhani and Vondr\'ak showed the following result that a Sperner-admissible labeling of $H_{k,n}$ does not have too many monochromatic hyperedges, which also implies that the number of non-monochromatic hyperedges is large.

\begin{theorem}[Proposition 2.1 in \cite{MV15}] \label{thm:MV}
Let $\ell$ be a Sperner-admissible labeling of $\Delta_{k,n}$. Then, the number of monochromatic hyperedges in $H_{k,n}$ is at most $\binom{n+k-3}{k-1}$, and therefore the number of non-monochromatic hyperedges is at least $\binom{n+k-3}{k-2}$.
\end{theorem}

Our main result of this section is an extension of the above result to the case when there are some inadmissible labels on a single face of $\Delta_{k,n}$. We show that a labeling in which all inadmissible labels are on a single face still has a large number of non-monochromatic hyperedges. We will denote the nodes $x\in \Delta_{k,n}$ with $\supp(x)\subseteq [k-1]$ as $\face(s_1,\ldots,s_{k-1})$.

\begin{theorem} \label{thm:nonmon}
Let $\ell$ be a labeling of $\Delta_{k,n}$ such that
all inadmissible labels are on $\face(s_1,\ldots, s_{k-1})$ and the number of nodes with inadmissible labels is $\beta \frac{(n+k-2)!}{n!}$ for some $\beta$.
Then, the number of non-monochromatic hyperedges of $H_{k,n}$ is at least
\[
\left(\frac{1}{(k-2)!}-\beta\right)\frac{(n+k-3)!}{(n-1)!}.
\]
\end{theorem}

\begin{proof}
Let $Z:= \{x \in \face(s_1,\ldots,s_{k-1}): \ell(x) = k\}$, i.e.\ $Z$ is the set of nodes in $\face(s_1,\ldots,s_{k-1})$ having an inadmissible label. Let us call a hyperedge of $H_{k,n}$ \emph{inadmissible} if the label of one of its nodes is inadmissible.



\begin{claim}
There are at most $\beta\frac{(n+k-3)!}{(n-1)!}$ inadmissible monochromatic hyperedges.
\end{claim}


\begin{proof}
Let $\mathcal{E}'$ be the set of inadmissible monochromatic hyperedges.
Each hyperedge $e\in \mathcal{E}'$ has exactly $k-1$ nodes from $\face(s_1,\ldots,s_{k-1})$ and they all have the same label as $e$ is monochromatic. Thus, each $e\in \mathcal{E}'$ contains $k-1$ nodes from $Z$. We define an injective map $\varphi:\mathcal{E}' \to Z$ by letting $\varphi(e)$ to be the node $x\in e\cap Z$ with the largest $1$st coordinate.
Notice that if $x=\varphi(e)$, then the other nodes of $e$ are $x-(1/n)e_1+(1/n)e_i$ ($i=2,\dots,k$), and all but the last one are in $Z$. In particular, $x_1$ is positive.

Let $Z' \subseteq Z$ be the image of $\varphi$. For $x \in Z$ and $i \in \{2,\dots,k-1\}$, let
\[
Z_x^i:=\{y \in Z: y_j=x_j\ \forall j \in [k-1]\setminus \{1,i\}\}.
\]
Since $y_k=0$ and $\|y\|_1=1$ for every $y \in Z$, the nodes of $Z_x^i$ are on a line containing $x$. It also follows that
$Z_x^i \cap Z_x^j=\{x\}$ if $i \neq j$. Let
\[
Z'':=\{x \in Z: \exists i \in \{2,\dots, k-1\}\ \text{such that}\ x_i \geq y_i\ \forall y \in Z_x^i \}.
\]
We observe that if $x\in Z'$, then for each $i\in \{2,\ldots, k-1\}$, the node $y=x-(1/n)e_1+(1/n)e_i$ is in $Z$ and hence, $y\in Z_x^i$ with $y_i>x_i$.
In particular, this implies that $Z'\cap Z''=\emptyset$.
We now compute an upper bound on the size of $Z\setminus Z''$, which gives an upper bound on the size of $Z'$ and hence also on the size of $\mathcal{E}'$, as $|Z'|=|\mathcal{E}'|$. For each node $x \in Z \setminus Z''$ and for every $i \in \{2,\dots,k-1\}$, let $z_x^i$ be the node in $Z'' \cap Z_x^i$ with the largest $i$th coordinate. Clearly $z_x^i \neq z_x^j$ if $i \neq j$, because $Z_x^i \cap Z_x^j=\{x\}$.

For given $y\in Z''$ and $i\in \{2,\ldots, k-1\}$, we want to bound the size of $S:=\{x\in Z\setminus Z'':z_x^i=y\}$. Consider $a\in S$. Then, $z_a^i=y$ implies that the node in $Z''\cap Z_a^i$ with the largest $i$-th coordinate is $y$. That is, $y_j=a_j$ for all $j\in [k-1]\setminus \{1,i\}$ and moreover $y_i\ge a_i$.
If $y_i=a_i$, then $y=a$, so $a$ is in $Z''$ which contradicts $a\in S$. Thus, $y_i>a_i$ for any $a \in S$, i.e.\ the nodes in $S$ are on the line $Z_y^i$ and their $i$-th coordinate is strictly smaller than $y_i$. This implies that $|S| \leq ny_i$. Consequently, the size of the set $\{x \in Z \setminus Z'': y=z_x^i \text{ for some } i\in \{2,\ldots, k-1\}\}$ is at most $n$, since $\sum_{i=2}^{k-2} y_i \leq \|y\|_1=1$.

For each $x\in Z\setminus Z''$, we defined $k-2$ distinct nodes $z_x^2,\ldots, z_x^{k-1}\in Z''$. Moreover, for each $y\in Z''$, we have at most $n$ distinct nodes $x$ in $Z\setminus Z''$ for which there exists $i\in \{2,\ldots, k-1\}$ such that $y=z_x^i$. Hence, $(k-2)|Z\setminus Z''|\le n|Z''|$, and therefore $|Z\setminus Z''|\le (n/(n+k-2))|Z|$.
This gives
\[|\mathcal{E}'|=|Z'| \leq |Z \setminus Z''|\leq \frac{n}{n+k-2}|Z| \leq \beta \frac{(n+k-2)!}{n!}\frac{n}{n+k-2}
=\beta\frac{(n+k-3)!}{(n-1)!},\]
as required.
\end{proof}

Let $\ell'$ be a Sperner-admissible labeling obtained from $\ell$ by changing the label of each node in $Z$ to an arbitrary admissible label. By Theorem \ref{thm:MV}, the number of monochromatic hyperedges for $\ell'$ is at most $\binom{n+k-3}{k-1}$. By combining this with the claim, we get that the number of monochromatic hyperedges for $\ell$ is at most $\binom{n+k-3}{k-1}+\beta\frac{(n+k-3)!}{(n-1)!}$. Since $H_{k,n}$ has $\binom{n+k-2}{k-1}$ hyperedges, the number of non-monochromatic hyperedges is at least
\begin{align*}
\binom{n+k-2}{k-1}-\binom{n+k-3}{k-1}-\beta\frac{(n+k-3)!}{(n-1)!}
&=\binom{n+k-3}{k-2}-\beta\frac{(n+k-3)!}{(n-1)!}\\
&= \left(\frac{1}{(k-2)!}-\beta \right)\frac{(n+k-3)!}{(n-1)!}.
\end{align*}
\end{proof}
We note that Theorem \ref{thm:nonmon} is tight for the extreme cases where $\beta=0$ and $\beta=1/(k-2)!$.  

We now derive Lemma \ref{lemma:non-opposite-cuts-in-discretezed-simplex} from Theorem \ref{thm:nonmon}. We restate Lemma \ref{lemma:non-opposite-cuts-in-discretezed-simplex} for convenience.
\lemmaNonOppCutSize*
\begin{proof}[Proof of Lemma~\ref{lemma:non-opposite-cuts-in-discretezed-simplex}]

Let $\ell$ be the labeling of $\Delta_{4,n}$ obtained from $P$ by setting $\ell(x)=4$ if $P(x)=5$, and $\ell(x)=P(x)$ otherwise. This is a labeling with $(\frac{1}{2}-\alpha) (n+1)(n+2)$ nodes having an inadmissible label, all on $\face(s_1,s_2,s_3)$. We apply Theorem \ref{thm:nonmon} with parameters $k=4$, $\beta=\frac{1}{2}-\alpha$, and the labeling $\ell$.
By the theorem, the number of non-monochromatic hyperedges in $H_{4,n}=(\Delta_{4,n},\mathcal{E})$ under labeling $\ell$ is at least $\alpha n(n+1)$.

We observe that for each hyperedge $e=\{u_1,u_2,u_3,u_4\}\in \mathcal{E}$, the subgraph $G[e]$ induced by the nodes in $e$ contains $6$ edges. Also, for any two hyperedges $e_1$ and $e_2$, the edges in the induced subgraphs $G[e_1]$ and $G[e_2]$ are disjoint as $e_1$ and $e_2$ can share at most one node. Moreover, for each non-monochromatic hyperedge $e\in \mathcal{E}$, at least $3$ edges of $G[e]$ are in $\delta(P)$. Thus, the number of edges of $G$ that are in $\delta(P)$ is at least $3\alpha n (n+1)$.

\end{proof} 
\section{Properties of the $3$-way cut instance in \cite{AMM17}}\label{sec:3-way-instance-properties}
In this section, we prove Lemma \ref{lemma:3-way-instance-properties} and Corollary \ref{coro:red-island-cut-cost} which are properties of the gap instance in \cite{AMM17}.

\subsection{The gap instance in \cite{AMM17}}
In this section, we summarize the relevant background about the gap instance against non-opposite $3$-way cuts designed by Angelidakis, Makarychev and Manurangsi \cite{AMM17}. For our purposes, we scale the costs of their instance by a factor of $6/5$ as it will be convenient to work with them. We describe this scaled instance now.

Let $\mathcal{G} = (\VFace,\EFace)$ where $\EFace:=\{xy:x,y\in \Delta_{3,n}, \|x-y\|_1=2/n\}$.
Their instance is obtained by dividing $\Delta_{3,n}$ into a middle hexagon $H:=\{x\in \Delta_{3,n}:x_i\le 2/3\ \forall\ i\in [3]\}$ and three corner triangles $T_1, T_2, T_3$, where $T_i:=\{x\in \Delta_{3,n}: x_i>2/3\}$. To define the edge costs, we let $\rho:=3/(5n)$.
The cost of the edges in $\mathcal{G}[H]$ is $\rho$. The cost of the non-boundary edges in $\mathcal{G}[T_i]$ that are not parallel to the opposite side of $e^i$ is also $\rho$. The cost of the non-boundary edges in $\mathcal{G}[T_i]$ that are parallel to the opposite side of $e_i$ are zero. The cost of the boundary edges in $L_{ij}$ are as follows: the edge closest to $e^i$ has cost $(n/3)\rho$, the second closest edge to $e^i$ has cost $(n/3-1)\rho$, and so on. See Figure \ref{fig:AMM-instance} for an example. We will denote the resulting graph with edge-costs as $J$.
The cost of a subset $F$ of edges on the instance $J$ is $Cost_J(F):=\sum_{e \in F} w(e)$.

\begin{figure}[htb]
\centering
\includegraphics[width=0.4\textwidth]{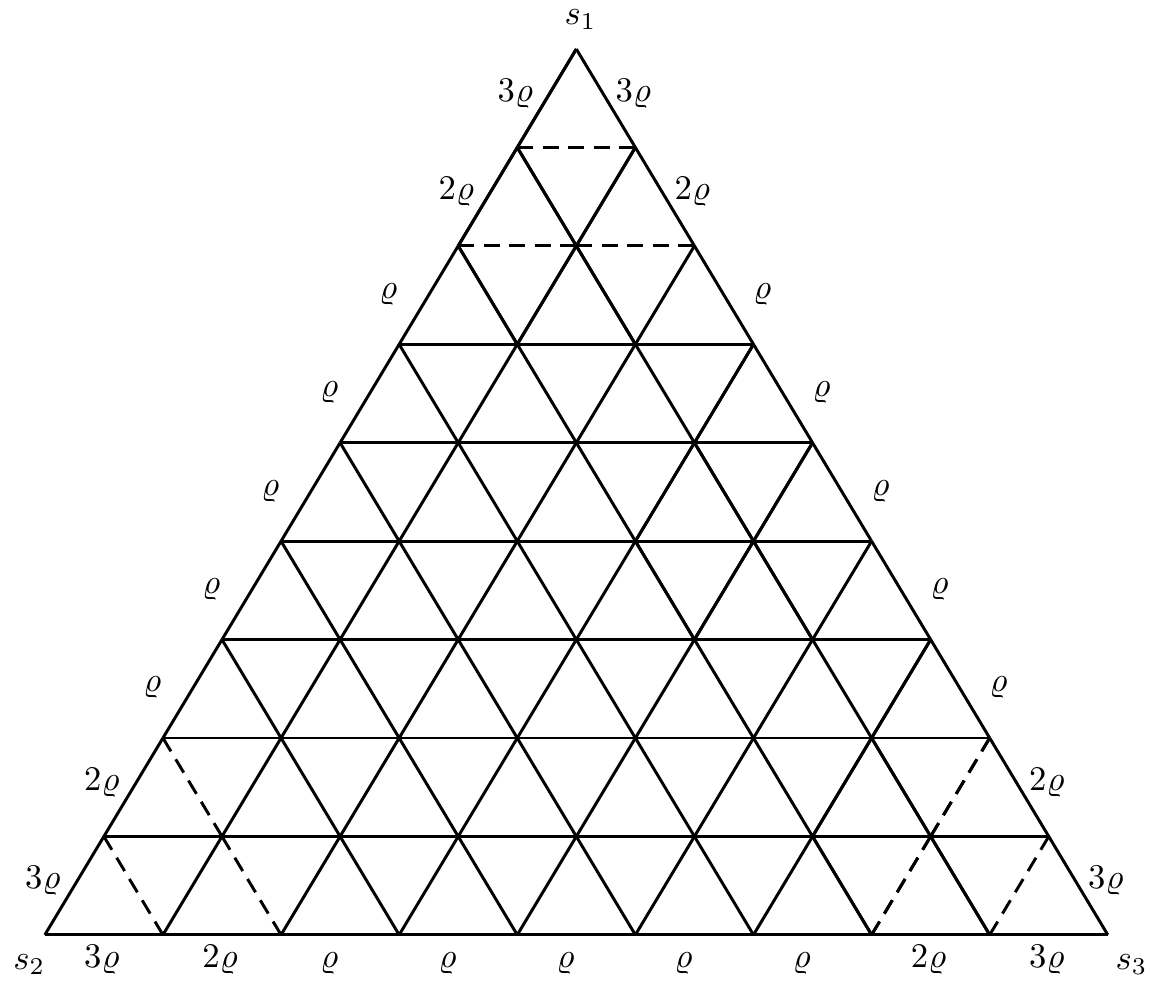}
\caption{The instance in \cite{AMM17} for $n=9$.}
\label{fig:AMM-instance}
\end{figure}

For a subset of edges $F \subset \EFace$, let $\mathcal{G}-F$ denote the graph $(\VFace,\EFace\setminus F)$.
We need the following two results about their instance. The first result shows that the cost of non-opposite cuts on their instance is at least $1.2$.
\begin{lemma}\label{lemma:non_opposite_cut_cost}\cite{AMM17}
For every non-opposite cut $Q:\VFace\rightarrow [4]$, the cost of $Q$ on instance $J$ is at least $1.2-\errorterm$.
\end{lemma}

The second result shows that if we remove a set of edges to ensure that a terminal $s_i$ cannot reach any node in the opposite side $V_{jk}$, then the cost of such a subset of edges is at least $0.4$.
\begin{lemma}\label{lemma:si_V_jk_cut}
For $\{i,j,k\}=[3]$ and for every subset $F$ of edges in $\EFace$
such that $s_i$ cannot reach $V_{jk}$ in $\mathcal{G}-F$, the cost of $F$ on instance $J$ is at least $0.4-(\errorterm)/3$.
\end{lemma}
Although Lemma \ref{lemma:si_V_jk_cut} is not explicitly stated in \cite{AMM17}, its proof appears under Case $1$ in the Proof of Lemma $3$ of \cite{AMM17}. The factor $0.4$ that we have here is because we scaled their costs by a factor of $6/5$.


We next define non-oppositeness as a property of the cut-set as it will be convenient to work with this property for cut-sets rather than for cuts.
\begin{definition}\label{defn:non_opposite}
A set $F \subseteq \EFace$ of edges is a non-opposite cut-set if there is no path from $s_1$ to $V_{23}$ in $\mathcal{G}-F$, no path from $s_2$ to $V_{13}$ in $\mathcal{G}-F$, and no path from $s_3$ to $V_{12}$ in $\mathcal{G}-F$.
\end{definition}

We summarize the connection between non-opposite cut-sets and non-opposite cuts.
\begin{prop}\label{prop:delta_Q_non_opposite}
\mbox{}

\begin{enumerate}[label=(\roman*)]
\item If $Q:\Delta_{3,n}\rightarrow [4]$ is a non-opposite cut, then $\delta(Q)$ is a non-opposite cut-set. \label{item:delta_Q_nonopp1}
\item For every non-opposite cut-set $F\subseteq \EFace$, the cost of $F$ on instance $J$ is at least $1.2-\errorterm$. \label{item:delta_Q_nonopp2}
\end{enumerate}
\end{prop}

\begin{proof}
\begin{enumerate}[label=(\roman*)]
\item Suppose not. Without loss of generality, suppose there exists a path from $s_1$ to $V_{23}$ in $\mathcal{G}-\delta(Q)$. Then, by the definition of $\delta(Q)$, all nodes of the path have the same label, so there exists a node $u\in V_{23}$ that is labeled as $1$, contradicting the fact that $Q$ is a non-opposite cut.
\item Consider a labeling $L:\Delta_{3,n}\rightarrow [4]$ where $L=i$ if the node $v$ is reachable from terminal $s_i$ in $\mathcal{G}-F$ and $L(v)=4$ if the node $v$ is reachable from none of the three terminals in $\mathcal{G}-F$. Since $F$ is a non-opposite cut-set, it follows that $\ell$ is a non-opposite cut. Moreover, $\delta(L)\subseteq F$.
Therefore, the claim follows by Lemma \ref{lemma:non_opposite_cut_cost}.
\end{enumerate}
\end{proof}

\subsection{Proof of Lemma \ref{lemma:3-way-instance-properties}}
We now restate and prove Lemma \ref{lemma:3-way-instance-properties}, i.e., a non-fragmenting non-opposite cut in $\Delta_{3,n}$ that has lot of nodes labeled as $4$ has large cost.
\lemmaThreeWayProps*

\begin{proof}
We first show that the labeling $Q$ may be assumed to indicate reachability in the graph $\mathcal{G}-\delta(Q)$.

\begin{claim}\label{claim:reachability-assumption}
For every non-opposite \noncornercut $Q:\Delta_{3,n}\rightarrow [4]$, there exists a labeling $Q':\Delta_{3,n}\rightarrow [4]$ such that
\begin{enumerate}[label=(\roman*)]
\item a node $v\in \VFace$ is reachable from $s_i$ in $\mathcal{G}-\delta(Q')$  iff $Q'(v) = i$,
\item $Cost_J(\delta(Q')) \leq Cost_J(\delta(Q))$,
\item the number of nodes in $\VFace$ that are labeled as $4$ by $Q$ is at most the number of nodes in $\VFace$ that are labeled as $4$ by $Q'$, and
\item $Q'$ is a non-opposite \noncornercut.
\end{enumerate}
\end{claim}

\begin{proof}
\begin{figure}[htb]
\centering
\includegraphics{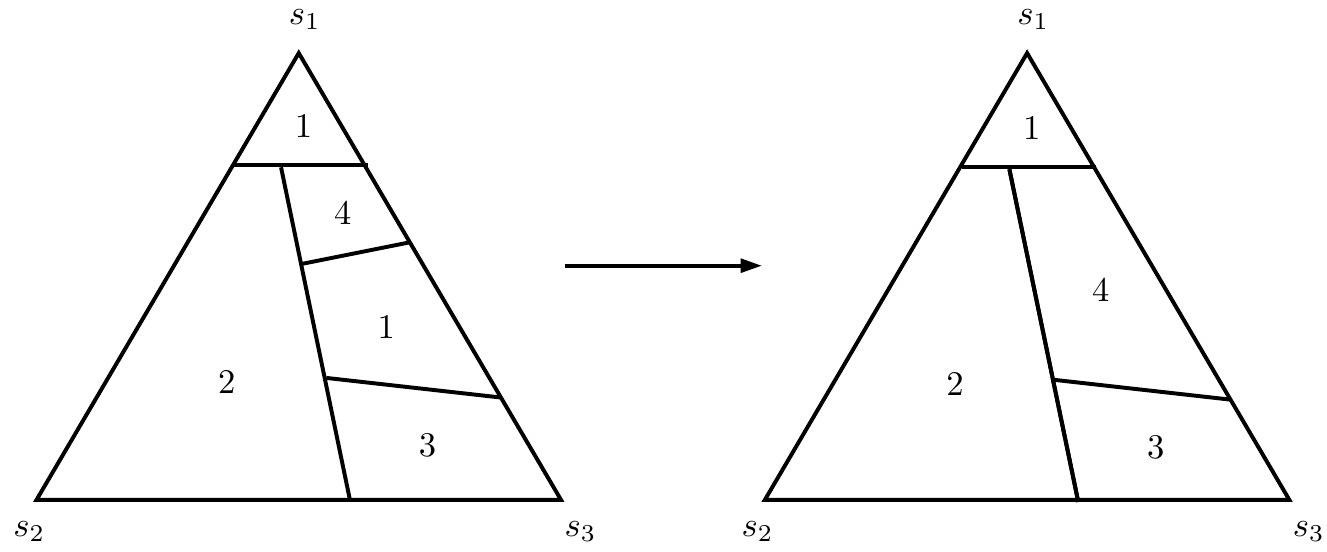}
\caption{An example of a cut $Q$ and the cut $Q'$ obtained in the proof of Claim \ref{claim:reachability-assumption}.}
\label{fig:Q-to-Q-prime}
\end{figure}


For $i \in [3]$, let $S_i$  be the set of nodes that can be reached from $s_i$ in $\mathcal{G}-\delta(Q)$. Consider a labeling $Q'$ defined by
\[
Q'(v) :=
\begin{cases}
i &\text{ if }v \in S_i, \text{ and}\\
4 &\text{ if }v \in \VFace \setminus (S_1 \cup S_2 \cup S_3).
\end{cases}
\]
See Figure \ref{fig:Q-to-Q-prime} for an example of a cut $Q$ and the cut $Q'$ obtained as above. We prove the required properties for the labeling $Q'$ below.

\begin{enumerate}[label=(\roman*)]
\item By definition, $Q'(v) = i$ iff $v$ is reachable from $s_i$ in $\mathcal{G} \setminus \delta(Q')$.
\item Since $\delta(Q')=(\delta(S_1) \cup \delta(S_2) \cup \delta(S_3)) \cap \delta(Q)$, we have that $\delta(Q') \subseteq \delta(Q)$. Hence, $Cost_J(\delta(Q'))\leq Cost_J(\delta(Q))$.
\item Let $i\in [3]$. Since all nodes of $S_i$ are labeled as $i$ by $Q$, the nodes labeled as $i$ by $Q'$ is a subset of the set of nodes labeled as $i$ by $Q$. This implies that $Q'$ is also a non-opposite cut and that the number of nodes in $\VFace$ that are labeled as $4$ by $Q$ is at most the number of nodes in $\VFace$ that are labeled as $4$ by $Q'$.
\item Since $Q$ is a \noncornercut, there exist distinct $i, j \in [3]$ such that $|\delta(Q) \cap L_{ij}|=1$.  Since $\delta(Q') \subseteq \delta(Q)$, we have that $|\delta(Q') \cap L_{ij}| \leq 1$. On the other hand, $Q'$ labels $s_i$ by $i$ and $s_j$ by $j$ and hence, $|\delta(Q') \cap L_{ij}| \geq 1$. Combining the two, we have that $|\delta(Q') \cap L_{ij}| =1$ and hence $Q'$ is a \noncornercut.
\end{enumerate}
\end{proof}

Let $\mathcal{G}':=\mathcal{G}-\delta(Q)$. By Claim \ref{claim:reachability-assumption}, we may henceforth assume that
\begin{equation}
\text{For every node $v\in V$, $v$ is reachable from $s_i$ in $\G$ iff $Q(v)=i$}.
\end{equation}
In order to show a lower bound on the cost of $Q$, we will modify $Q$ to obtain a non-opposite cut while reducing its cost by at least $\ctwo\alpha$. \cite{AMM17} showed that the cost of every non-opposite cut on $J$ is at least $1.2-\errorterm$. Therefore, the cost of $Q$ on $J$ must be at least $1.2-\errorterm + \ctwo\alpha$.

Since $Q$ is a \noncornercut, there exist distinct $i,j \in [3]$ such that $|\delta(Q) \cap L_{ij}| = 1$. Without loss of generality, suppose that $i = 1$ and $j = 3$.
For $i \in [3]$, let $S_i:=\{v \in \VFace \mid Q(v) = i\}$, i.e.\  $S_i$  is the set of nodes that can be reached from $s_i$ in $\mathcal{G}'$.
Let $B := \{v \in \VFace \mid Q(v) = 4\}$ be the set of nodes labeled as $4$ by $Q$. Then, $|B| = \alpha n^2$. We note that $S_1$, $S_2$, and $S_3$ are components of $\G$, and
the set $B$ is the union of the remaining components.

We recall that $V_{ij}$ is the set of end nodes of edges in $L_{ij}$.
We say that a node $v \in \VFace$ can reach $V_{ij}$ in $\G$ if there exists a path from $v$ to some node $w \in V_{ij}$ in $\G$.
We observe that all nodes in $V_{13}$ are reachable from either $s_1$ or $s_3$ in $\G$. In particular, this means that no node of $B$ can reach $V_{13}$ in $\G$. We partition the node set $B$ based on reachability as follows (see Figure \ref{fig:partition-of-B}):
\begin{align*}
B_1 &:= \{v \in B \mid \text{  $v$ cannot reach $V_{12}$ and $V_{23}$ in $\G$}\}, \\
B_2 &:= \{ v \in B \mid \text{ $v$ can reach $V_{12}$ but not $V_{23}$ in $\G$}\}, \\
B_3 &:= \{ v \in B \mid  \text{ $v$ can reach $V_{23}$ but not $V_{12}$ in $\G$}\}, \text{ and}\\
B_4 &:= \{v \in B \mid \text{$v$ can reach $V_{12}$ and $V_{23}$ in $\G$}\}.
\end{align*}

\begin{figure}[htb]
\centering
\includegraphics{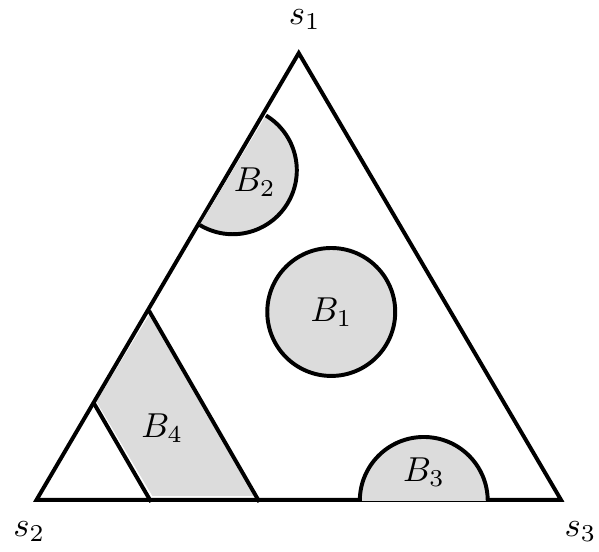}
\caption{Partition of $B$ into $B_1, B_2, B_3, B_4$.}
\label{fig:partition-of-B}
\end{figure}

For $r \in [4]$, let $\beta_r := |B_r|/n^2$. We next summarize the properties of the sets defined above.
\begin{prop}\label{prop:B_i_properties}
The sets $B_1, B_2, B_3, B_4$ defined above satisfy the following properties:
\begin{enumerate}[label=(\roman*)]
\item For every distinct $r,p\in [4]$, we have $B_r\cap B_p=\emptyset$. \label{item:Bprop1}
\item For every $r\in [4]$, we have $\delta(B_r) \subseteq \delta(Q)$, i.e.\ $B_r$ is the union of some components of $\G$. \label{item:Bprop2}
\item For every $r\in [4]$ and every edge $e \in \delta(B_r)$, one end node of $e$ is in $B_r$ and the other one is in $S_1\cup S_2\cup S_3$. \label{item:Bprop3}
\item For every distinct $r,p\in [4]$, we have $\delta(B_r) \cap \delta(B_p) = \emptyset$.\label{item:Bprop4}

\item $B = \cup_{r=1}^4 B_r, \sum_{r=1}^4 \beta_r = \alpha,$ and $\beta_r \leq \alphabound$ for every $r \in [4]$. \label{item:Bprop5}
\end{enumerate}
\end{prop}
\begin{proof}
\begin{enumerate}[label=(\roman*)]
\item The disjointness property follows from the definition of the sets.

\item Suppose $\delta(B_r)$ is not a subset of $\delta(Q)$ for some $r\in [4]$. Without loss of generality, let $r=1$ (the proof is similar for the other cases). Then, there exists an edge $uv\in \EFace\setminus \delta(Q)$ with $u\in B_1,v\in B\setminus B_1$. Since $v$ is in $B\setminus B_1$, it follows that the node $v$ can reach either $V_{12}$ or $V_{13}$ in $\G$. Moreover, since the edge $uv$ is in $\mathcal{G}'$, it follows that the node $u$ can also reach either $V_{12}$ or $V_{13}$ in $\G$, and hence $u\not\in B_1$. This contradicts the assumption that $u\in B_1$.

\item Let $uv\in \delta(B_r)$ with $u\in B_r$ and $v\not\in B_r$. Since $Q(u)=4$, the node $u$ is not reachable from any of the terminals in $\G$. Suppose that the node $v$ is also not reachable from any of the terminals in $\G$. Then, by the reachability assumption, it follows that $Q(v)=4$. Hence, the edge $uv$ has both end-nodes labeled as $4$ by $Q$ and therefore $uv\not\in \delta(Q)$. Thus, we have an edge $uv\in \delta(B_i)\setminus \delta(Q)$ contradicting part \ref{item:Bprop2}.

\item Follows from parts \ref{item:Bprop1} and \ref{item:Bprop3}.

\item By definition, we have that $B=\cup_{r=1}^4 B_r$. Since the sets $B_1,B_2,B_3,B_4$ are pair-wise disjoint, they induce a partition of $B$ and hence $|B|=\sum_{r=1}^4 |B_r|$. Consequently, $\sum_{r=1}^4 \beta_r n^2 =\alpha n^2$ and thus, $\sum_{r=1}^4 \beta_r = \alpha$. Next, we note that $|\VFace|=(n+1)(n+2)/2$. Since $B_r\subseteq B\subseteq \VFace$, we have that $\beta_r=|B_r|/n^2\le |\VFace|/n^2\le (1+1/n)(1+2/n)/2 \le 0.66$ since $n\ge 10$.
\end{enumerate}
\end{proof}


By Proposition \ref{prop:delta_Q_non_opposite} \ref{item:delta_Q_nonopp1}, the cut-set $\delta(Q)$ is a non-opposite cut-set.
The following claim shows a way to modify $\delta(Q)$ to obtain a non-opposite cut-set with strictly smaller cost if $\beta_r>0$.
\begin{claim}\label{claim:changing_Q} For every $r \in [4]$, there exists $E_r \subseteq \delta(B_r), E_r' \subseteq \mathcal{G}[B_r]$ such that
\begin{enumerate}
\item $E_r \subseteq \delta(S_i)$ for some $i \in [3]$,
\item $(\delta(Q) \setminus E_r) \cup E_r'$ is a non-opposite cut-set and
\item $Cost_J(E_r) - Cost_J(E_r') \geq \ctwo \beta_r$.
\end{enumerate}
\end{claim}

\begin{proof}
We consider the cases $r=1,2,4$ individually as the proofs are different for each of them. The case of $r=3$ is similar to the case of $r=2$. We begin with a few notations that will be used in the proof. For distinct $i,j\in [3]$, and for $t\in \{0,1,\ldots, 2n/3\}$, let $V_{ij}^t:=\{u\in \VFace:u_k=1-t/n\text{ for }\{k\}=[3]\setminus \{i,j\}\}$. Thus, $V_{ij}^t$ denotes the set of nodes that are on the line parallel to $V_{ij}$ and at distance $t/n$ from it. We will call the sets $V_{ij}^t$ as lines for convenience.
Let $L_{ij}^t$ denote the edges of $\EFace$ whose end-nodes are in $V_{ij}^t$. Thus, the edges in $L_{ij}^t$ are parallel to $L_{ij}$ (see Figure \ref{fig:L_ij_t}).

\begin{figure}[htb]
\centering
\includegraphics{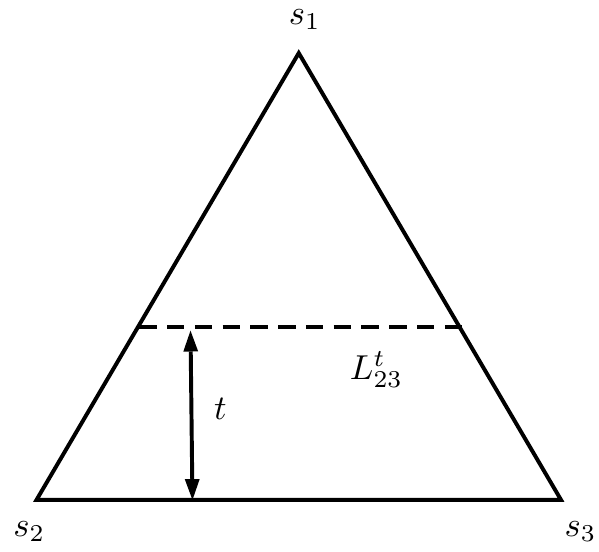}
\caption{The set of edges $L_{23}^t$.}
\label{fig:L_ij_t}
\end{figure}

\begin{enumerate}[leftmargin=*]
\item \textbf{Suppose $r=1$.}
We partition the set $\delta(B_1)$ of edges into three sets $X_i := \delta(B_1)\cap \delta(S_i)$ for $i \in [3]$ (see Figure \ref{fig:delta_B_1_partition}).

\begin{figure}[htb]
\centering
\includegraphics{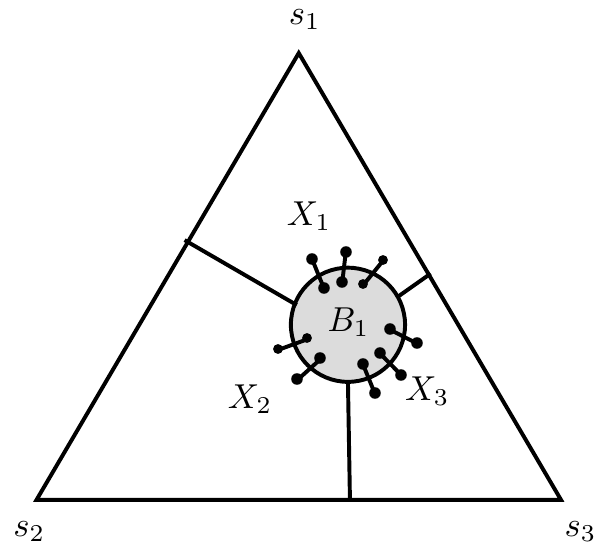}
\caption{Partition of $\delta(B_1)$ into $X_i$'s.}
\label{fig:delta_B_1_partition}
\end{figure}



By Proposition \ref{prop:B_i_properties} \ref{item:Bprop3}, we have that $(X_1,X_2,X_3)$ is a partition of $B_1$.
Let
\begin{align*}
E_1&:=\arg\max\{Cost_J(F):F\in \{X_1,X_2,X_3\}\} \text{ and }\\
E_1'&:=\emptyset.
\end{align*}
We now show the required properties for this choice of $E_1$ and $E_1'$.
\begin{enumerate}
\item Since $E_1'=\emptyset$, we need to show that $\delta(Q)\setminus E_1$ is a non-opposite cut-set. Let $\mathcal{G}'':=\mathcal{G}-(\delta(Q)\setminus E_1)$.
For each edge $e \in E_1$, the end node of $e$ in $\VFace\setminus B_1$ is reachable from a terminal $s_i$ in $\G$ iff it is reachable from $s_i$ in $\mathcal{G}''$. Therefore, for each node $v\in \VFace\setminus B_1$ and a terminal $s_i$ for $i\in [3]$, we have that $v$ is reachable from $s_i$ in $\mathcal{G}'$ iff $v$ is reachable from $s_i$ in $\mathcal{G}''$. Since $\delta(Q)$ is a non-opposite cut-set, it follows that $s_i$ cannot reach $V_{jk}$ in $\mathcal{G}'$ for $\{i,j,k\}=[3]$. Since $B_1\cap (V_{12}\cup V_{23}\cup V_{13})=\emptyset$, the terminal $s_i$ cannot reach $V_{jk}$ in $\mathcal{G}''$ for $\{i,j,k\}=[3]$. Hence, $\delta(Q)\setminus E_1$ is a non-opposite cut-set.

\item We note that none of the nodes in $B_1$ can reach $V_{12}$, $V_{23}$ and $V_{13}$ in $\G$. Therefore, if there exists a node from $B_1$ in $V_{ij}^t$ for some $t\in \{1,\ldots, n\}$, then at least two edges in $L_{ij}^t$ should be in $\delta(B_1)$ (see Figure \ref{fig:B_1_L_23t}). Therefore, if $V_{ij}^t\cap B_1\neq \emptyset$, then $|\delta(B_1)\cap L_{ij}^t|\ge 2$.

\begin{figure}[htb]
\centering
\includegraphics{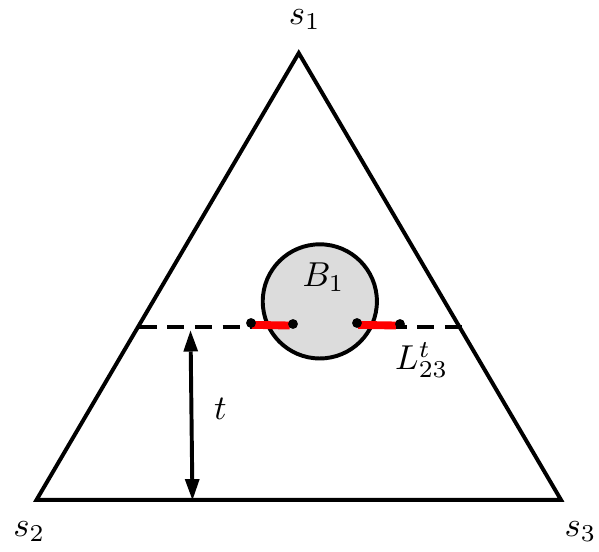}
\caption{$B_1 \cap V_{23}^t \neq \emptyset$ implies that $|\delta(B_1) \cap L_{23}^t|\geq 2$.}
\label{fig:B_1_L_23t}
\end{figure}
Every node $v\in B_1$ is in at least two lines among $V_{ij}^t$ for distinct $i,j\in [3]$ and $t\in \{1,\ldots, 2n/3\}$. Each line $V_{ij}^t$ for $t\in \{1,\ldots,2n/3\}$ has at most $n$ nodes. Hence, the number of lines with non-empty intersection with $B_1$ is at least $2|B_1|/n$. For each line that has a non-empty intersection with $B_1$, we have at least two edges in $\delta(B_1)$. Hence,
\[
\left|\delta(B_1) \cap \left(\cup_{i,j \in [3], t \in \{1,\ldots, 2n/3\}}  L_{ij}^t\right)\right| \geq 4 \cdot \frac{|B_1|}{n}.
\]
The cost of each edge in $\cup_{i,j \in [3], t \in \{1,\ldots, 2n/3\}} L_{ij}^t$ is $3/(5n)$. So,
\[
Cost_J(\delta(B_1)) \geq Cost_J\left(\delta(B_1) \cap \left(\cup_{i,j \in [3], t \in \{1,\ldots, 2n/3\}}  L_{ij}^t\right)\right)\ge \frac{12}{5} \frac{|B_1|}{n^2} = \frac{12}{5} \beta_1.
\]
Since we set $E_1$ to be the $X_i$ with maximum cost, we get that $Cost_J(E_1) \geq (4/5) \beta_1$. Moreover, $Cost_J(E_1') = 0$ as $E_1' = \emptyset$. Hence, $Cost_J(E_1) - Cost_J(E_1') \geq (4/5)\beta_1 \geq \ctwo\beta_1$.
\end{enumerate}

\item \textbf{Suppose $r=2$.}
We assume that $B_2\neq \emptyset$ as otherwise, the claim is trivial. Similar to the previous case, we partition the set $\delta(B_2)$ into three sets $X_i := \delta(B_2)\cap \delta(S_i)$ for $i \in [3]$ (see Figure \ref{fig:E_2_X_i}).
\begin{figure}[htb]
\centering
\includegraphics{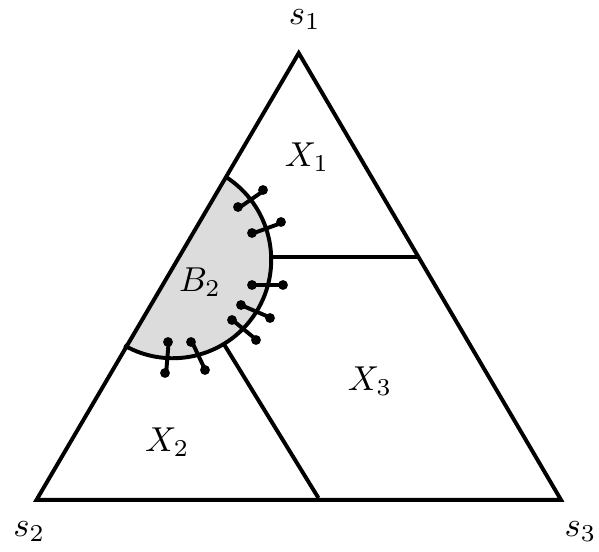}
\caption{Partition of $\delta(B_2)$ into $X_i$'s.}
\label{fig:E_2_X_i}
\end{figure}

We also define
\[
Z:=X_3\cap \delta(B_2\cap V_{12})
\]
and let
\begin{align*}
E_2 &:= X_1\text{ and } E_2':= \emptyset \text{ if $Cost_J(X_1)\ge \ctwo \beta_2$},\\
E_2 &:= X_2\text{ and } E_2':= \emptyset \text{ if $Cost_J(X_2)\ge \ctwo \beta_2$},\\
E_2 &:= X_3\setminus Z\text{ and } E_2':=\delta_\mathcal{G}(B_2\setminus V_{12},B_2\cap V_{12}) \text{ if $Cost_J(X_1), Cost_J(X_2)< \ctwo \beta_2$}.
\end{align*}
We emphasize that the last case is the only situation where we use a non-empty set for $E_2'$.
We now show the required properties for this choice of $E_2$ and $E_2'$.
\begin{enumerate}
\item Let $\mathcal{G}'':=\mathcal{G}-((\delta(Q)\setminus E_2)\cup E_2')$. For each edge $e\in E_2$, the end node of $e$ in $\VFace\setminus B_2$ is reachable from a terminal $s_i$ in $\mathcal{G}'$ iff it is reachable from $s_i$ in $\mathcal{G}''$. Therefore, for each node $v\in \VFace\setminus B_2$ and a terminal $s_i$ for $i\in [3]$, we have that $v$ is reachable from $s_i$ in $\mathcal{G}'$ iff $v$ is reachable from $s_i$ in $\mathcal{G}''$. Since $B_2\cap V_{13}=\emptyset$ and $s_2$ cannot reach $V_{13}$ in $\mathcal{G}'$, we have that $s_2$ cannot reach $V_{13}$ in $\mathcal{G}''$. Similarly, $s_1$ cannot reach $V_{23}$ in $\mathcal{G}''$. It remains to argue that $s_3$ cannot reach $V_{12}$ in $\mathcal{G}''$. We have two cases.

\begin{enumerate}
\item
Suppose $E_2=X_1$ or $E_2=X_2$. We note that $X_1$ and $X_2$ are the set of edges in $\delta(B_2)$ whose end nodes outside $B_2$ are reachable from $s_1$ (and $s_2$ respectively) in $\mathcal{G}'$. So, if $E_2=X_1$ or if $E_2=X_2$, then the set of nodes reachable by $s_3$ in $\mathcal{G}'$ and $\mathcal{G}''$ remains the same. Since $s_3$ cannot reach $V_{12}$ in $\mathcal{G}'$, we have that $s_3$ cannot reach $V_{12}$ in $\mathcal{G}''$.

\item Suppose $E_2=X_3$. We will show that $\delta(B_2\cap V_{12})\subseteq (\delta(Q)\setminus E_2)\cup E_2'$. Consequently, the nodes of $B_2\cap V_{12}$ are not reachable from $s_3$ in $\mathcal{G}''$. Since nodes of $V_{12}\setminus B_2$ are not reachable from $s_3$ in $\mathcal{G}'$, we have that $s_3$ cannot reach $V_{12}$.

We now show that $\delta(B_2\cap V_{12})\subseteq (\delta(Q)\setminus E_2)\cup E_2'$. Let $uv\in \delta(B_2\cap V_{12})$ with $u\in B_2\cap V_{12}$ and $v\not\in B_2\cap V_{12}$. If $v \in S_1 \cup S_2$, then $uv \in \delta(B_2)\subseteq \delta(Q)$ and $uv\not\in X_3\supseteq E_2$. Hence, $uv\in (\delta(Q)\setminus E_2)\cup E_2'$. If $v \in S_3$, then $uv\in Z$ and hence $uv\not\in E_2$. Moreover, $uv\in \delta(B_2)\subseteq \delta(Q)$, hence $uv\in (\delta(Q)\setminus E_2)\cup E_2'$. If $v\in B$, then $v\in B_2$ by Proposition \ref{prop:B_i_properties} \ref{item:Bprop4} and hence $e\in E_2'\subseteq (\delta(Q)\setminus E_2)\cup E_2'$.
\end{enumerate}
\item If $Cost_J(X_1)$ or $Cost_J(X_2)$ is at least $\ctwo \beta_2$, then we are done. So, let us assume that $Cost_J(X_1),Cost_J(X_2)\le \ctwo \beta_2$. Let $Y_1$, $Y_2$ and $Y_3$ be the set of edges in $\delta(B_2)$ that are parallel to $L_{12}, L_{13}$ and $L_{23}$ respectively (see Figure \ref{fig:E_2_Y_i}). Formally,
\begin{align*}
Y_1 &:= \delta(B_2)\cap \left(\cup_{t\in \{0,1,\ldots, n\}}L_{12}^t\right)\\
Y_2 &:= \delta(B_2)\cap \left(\cup_{t\in \{0,1,\ldots, n\}}L_{13}^t\right)\\
Y_3 &:= \delta(B_2)\cap \left(\cup_{t\in \{0,1,\ldots, n\}}L_{23}^t\right)
\end{align*}

\begin{figure}
\centering
\hspace{10mm}\includegraphics[width=0.9\textwidth]{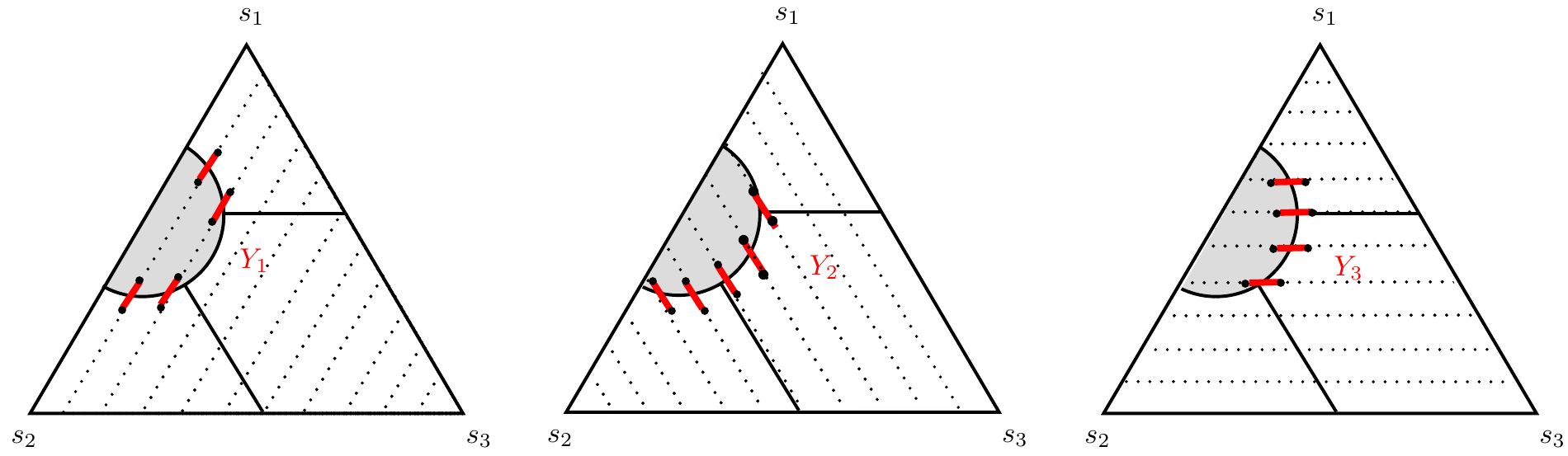}
\caption{Partition of $\delta(B_2)$ into $Y_i$'s. The shaded region is $B_2$.}
\label{fig:E_2_Y_i}
\end{figure}

Claims \ref{claim:Etwo-prime-cost} and \ref{claim:Yone-cost} will help us derive the required inequality on the cost.
\begin{claim}\label{claim:Etwo-prime-cost}
\[
Cost_J(E_2')\le Cost_J(Y_2)+Cost_J(Y_3)-Cost_J(Z).
\]
\end{claim}
\begin{proof}
We proceed in two steps: (1) we will show a one-to-one mapping $f$ from edges in $E_2'$ to edges in $(Y_2\cup Y_3)\setminus Z$ such that
the cost of every edge $e\in E_2'$ is the same as the cost of the mapped edge $f(e)$ in the instance $J$, i.e., $w(e)=w(f(e))$ for every $e\in E_2'$ and (2) we will show that that $Z\subseteq Y_2\cup Y_3$. Now, by observing that the sets $Y_2$ and $Y_3$ are disjoint, we get that $Cost_J(E_2')\le Cost_J(Y_2)+Cost_J(Y_3)-Cost_J(Z)$.

We now define the one-to-one mapping $f:E_2'\rightarrow (Y_2\cup Y_3)\setminus Z$. Let $e=uv\in E_2'$ such that $u\in B_2\cap V_{12}, v\in B_2\setminus V_{12}$.
Since $E_2'$ only contains edges between $B_2\cap V_{12}$ and $B_2\setminus V_{12}$, it does not contain an edge parallel to $L_{12}$.
Therefore, $e\in L_{13}^t$ or $e\in L_{23}^t$ for some $t\in \{1,\ldots, n\}$.
Suppose $e\in L_{13}^t$ for some $t\in \{1,\ldots, n\}$. Since the nodes of $B_2$ cannot reach $V_{23}$ in $\mathcal{G}'$, there exists an edge in $\delta(B_2)\cap L_{13}^t$. We map $e$ to an arbitrary edge in $\delta(B_2)\cap L_{13}^t\subseteq Y_2$ (see Figure \ref{fig:mapping_f}). We note that the set $Z$ contains the set of edges incident to $B_2\cap V_{12}$ whose other end node is in $S_3$. Since both $u$ and $v$ are in $B_2$, it follows that $L_{13}^t\cap Z=\emptyset$. So our mapping of $e$ is indeed to an edge in $Y_2\setminus Z$.
Similarly, if $e\in L_{23}^t$ for some $t\in \{1,\ldots, n\}$, then we map $e$ to an arbitrary edge in $\delta(B_2)\cap L_{23}^t\subseteq Y_3\setminus Z$. This mapping is a one-to-one mapping as $E_2'$ contains at most one edge from $L_{13}^t$ for each $t\in \{1,2,\ldots, n\}$ and at most one edge from $L_{23}^t$ for each $t\in \{1,2,\ldots, n\}$. Moreover, for each $t\in \{1,2,\ldots, n\}$, the cost of all edges in $L_{13}^t$ are identical and the cost of all edges in $L_{23}^t$ are identical.

\begin{figure}[htb]
\centering
\includegraphics{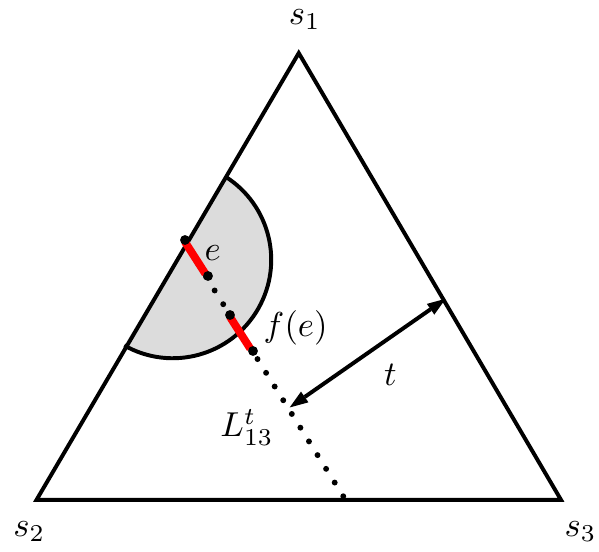}
\caption{Mapping from $E_2'$ to $(Y_2 \cup Y_3) \setminus Z$. The shaded region is $B_2$.}
\label{fig:mapping_f}
\end{figure}

We now show that $Z\subseteq Y_2\cup Y_3$. The set $Z$ contains all edges whose one end node is in $B_2\cap V_{12}$ and another end node is in $S_3$. Since $V_{12} \cap S_3=\emptyset$, the set $Z$ does not contain any edge between $B_2\cap V_{12}$ and $V_{12}\setminus B_2$. Hence, $Y_1\cap Z=\emptyset$. Since $Z$ is a subset of $X_3$ which is a subset of $Y_1\cup Y_2\cup Y_3$, it follows that $Z\subseteq Y_2\cup Y_3$.
\end{proof}

\begin{claim}\label{claim:Yone-cost}
\[
Cost_J(Y_1)\ge \frac{6}{5}\beta_2.
\]
\end{claim}
\begin{proof}
We first show a lower bound on the size of the set $W:=\{t\in \{0,1,\ldots, 2n/3\}:V_{12}^t\cap B_2\neq \emptyset\}$. If $B_2\cap V_{12}^t\neq \emptyset$ for some $t\in \{2n/3+1,\ldots, n\}$, then $B_2\cap V_{12}^t\neq \emptyset$ for all $t\in \{0,1,\ldots, 2n/3\}$ and hence, $|W|\ge 2n/3$. Otherwise, $B_2\cap V_{12}^t= \emptyset$ for all $t\in \{2n/3+1,\ldots, n\}$. In this case, $B_2\subseteq \cup_{t=0}^{2n/3}V_{12}^t$. For $t\ge 1$, each line $V_{12}^t$ has at most $n$ nodes. For $t=0$, the set $B_2$ can contain at most $n-1$ nodes from $V_{12}^0$ which are not $s_1$ or $s_2$. Hence, $|W|\ge |B_2|/n=\beta_2 n$. Thus, we have that $|W|\ge \min\{2n/3,\beta_2n\}=\beta_2n$ as $\beta_2\le 0.66$.

Since the nodes of $B_2$ cannot reach $V_{23}$ and $V_{13}$ in $\mathcal{G}'$, we have that $|\delta(B_2)\cap L_{12}^t|\ge 2$ if $B_2\cap V_{12}^t\neq \emptyset$. Hence,
\[
\left|\delta(B_2)\cap \left(\cup_{t=0}^{2n/3}L_{12}^t\right)\right|\ge 2|W|\ge 2\beta_2 n.
\]
Each edge in $\cup_{t=0}^{2n/3}L_{12}^t$ has cost at least $3/5n$. Hence,
\[
Cost_J\left(\delta(B_2)\cap \left(\cup_{t=0}^{2n/3}L_{12}^t\right)\right)\ge \frac{6}{5}\beta_2.
\]
Since $Y_1=\delta(B_2)\cap \left(\cup_{t=0}^{n}L_{12}^t\right)\supseteq \delta(B_2)\cap \left(\cup_{t=0}^{2n/3}L_{12}^t\right)$, we get that $Cost_J(Y_1)\ge (6/5)\beta_2$.
\end{proof}

We now derive the required inequality on the cost as follows:
\begin{align*}
&Cost_J(E_2) - Cost_J(E_2')  = Cost_J(X_3\setminus Z) - Cost(E_2')\\
&\quad \geq Cost_J(X_3) - Cost_J(Z) - Cost_J(Y_2) - Cost_J(Y_3) + Cost_J(Z) \quad \quad \text{(By Claim \ref{claim:Etwo-prime-cost})}\\
&\quad \ge Cost_J(X_3 \cap Y_1) + Cost_J(X_3 \cap Y_2) + Cost_J(X_3 \cap Y_3) - Cost_J(Y_2) - Cost_J(Y_3)\\
&\quad = Cost_J(X_3 \cap Y_1) - Cost_J((X_1\cup X_2) \cap Y_2) - Cost_J((X_1 \cup X_2) \cap Y_3)\\
&\quad = Cost_J(Y_1) - Cost_J( (X_1 \cup X_2)\cap Y_1) \\
&\quad\quad\quad\quad -Cost_J((X_1\cup X_2) \cap Y_2) - Cost_J((X_1 \cup X_2) \cap Y_3)\\
& \quad= Cost_J(Y_1) - Cost_J(X_1\cup X_2)\\
& \quad\geq \frac{6}{5}\beta_2 - \ctwo\beta_2 - \ctwo\beta_2 \quad \quad \quad \quad\quad\text{(By Claim \ref{claim:Yone-cost} and $Cost_J(X_1),Cost_J(X_2)\le \ctwo \beta_2$)}\\
& \quad = \ctwo \beta_2.
\end{align*}
\end{enumerate}

\item \textbf{Suppose $r=4$.}
We assume that $B_4\neq \emptyset$, as otherwise the claim is trivial. We partition $\delta(B_4)$ into
$X_1:=\delta(B_4) \cap \delta(S_2)$ and $X_2 := \delta(B_4)\setminus X_1$ (see Figure \ref{fig:B_4_X_1_X_2}),
and let $E_4:=X_1$ and $E_4':=\emptyset$.

We now show the required properties for this choice of $E_4$ and $E_4'$.
Let us fix a node $v\in B_4$ and a path $v, u_1,\ldots, u_t$ from $v$ to $L_{12}$ in $\mathcal{G}[B_4]$, and a path $v,w_1,\ldots, w_{t'}$ from $v$ to $L_{23}$ in $\mathcal{G}[B_4]$ (see Figure \ref{fig:B_4_X_1_X_2}). Let $S:=\{v,u_1,\ldots, u_t,w_1,\ldots, w_{t'}\}$. We note that $S\subseteq B_4$.

\begin{figure}[htb]
\centering
\includegraphics{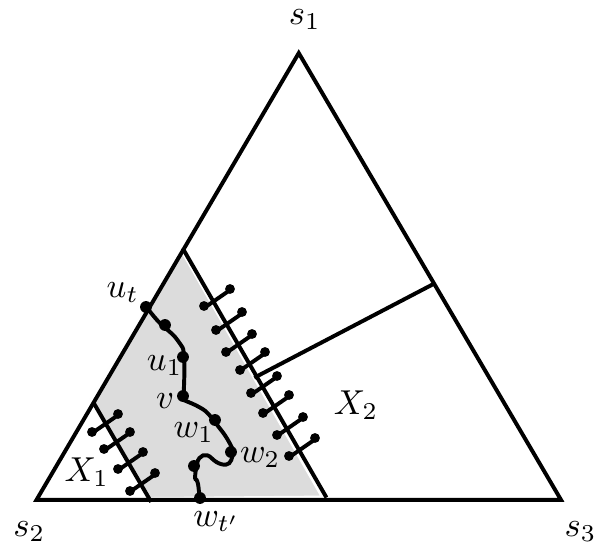}
\caption{Partition of $\delta(B_4)$ into $X_1$ and $X_2$. The shaded region is $B_4$.}
\label{fig:B_4_X_1_X_2}
\end{figure}


\begin{enumerate}
\item Since $E_4'=\emptyset$, we need to show that $\delta(Q)\setminus E_4$ is a non-opposite cut-set. Let $\mathcal{G}'':=\mathcal{G}-(\delta(Q)\setminus E_4)$.
We first observe that there are no paths between $S$ and $V_{13}$ in $\mathcal{G}-X_2$. Hence, there is no path from $s_1$ or $s_3$ to an end node of $E_4=X_1$ in $\mathcal{G}'$. Moreover, there is no path from $s_1$ to $V_{23}$ or from $s_3$ to $V_{12}$ in $\mathcal{G}'$. So, there is no path from $s_1$ to $V_{23}$ or from $s_3$ to $V_{12}$ in $\mathcal{G}''$. Also, since $X_2\subseteq \delta(Q)\setminus X_1$ and there is no path from $s_2$ to $V_{13}$ in $\mathcal{G}-X_2$, it follows that there is no path from $s_2$ to $V_{13}$ in $\mathcal{G}''$. Hence, $\delta(Q)\setminus E_4$ is a non-opposite cut-set.

\item We note that there are no paths between $s_2$ and $S$ in $\mathcal{G}-E_4$. Moreover, all paths in $\mathcal{G}$ between $s_2$ and $V_{13}$ go through $S$. Hence, there are no paths between $s_2$  and $V_{13}$ in $\mathcal{G}-E_4$. The cost of any such subset of nodes can be lower bounded by the Lemma \ref{lemma:si_V_jk_cut}. Thus, $Cost_J(E_4)-Cost_J(E_4')\ge 0.4-(\errorterm)/3\ge \ctwo\beta_4$. The last inequality is because $\beta_4\le 0.66$ by Proposition \ref{prop:B_i_properties} and $n\ge 10$.
\end{enumerate}
\end{enumerate}
\end{proof}


For $r\in [4]$, let $E_r$ and $E_r'$ be the sets given by Claim \ref{claim:changing_Q}. We will show that
\[
F:=\left(\delta(Q)\setminus \left(\cup_{r=1}^4 E_r\right)\right)\cup\left(\cup_{r=1}^4 E_r'\right)
\]
is a non-opposite cut-set and that $Cost_J(\delta(Q))\ge Cost_J(F)+\ctwo\alpha$. Then, we use Proposition \ref{prop:delta_Q_non_opposite} \ref{item:delta_Q_nonopp2}
to conclude that $Cost_J(\delta(Q))\ge 1.2-\errorterm + \ctwo\alpha$.


\begin{claim}
$F$ is a non-opposite cut-set.
\end{claim}
\begin{proof}
 Let $\mathcal{G}''=\mathcal{G}-F$, and for $i \in [3]$, let $S_i'$ be the set of nodes reachable from $s_i$ in $\mathcal{G}''$. Since $E_r' \subseteq \mathcal{G}[B]$ for every $r$,
$S_i'$ is a superset of $S_i$, and $\mathcal{G}''[S_i]=\G[S_i]$, which is connected. By the first property of Claim \ref{claim:changing_Q}, for every $r \in [4]$ there exists $i\in [3]$ such that $E_r \subseteq \delta(B_r) \cap \delta(S_i)$. This implies, together with Proposition \ref{prop:B_i_properties} \ref{item:Bprop2}, that the sets $S_i'$ are disjoint. It also implies the following property:
\begin{itemize}
\item[($\star$)] For every $r \in [4]$, there exists $i\in [3]$ such that $\delta_{\mathcal{G}''}(B_r) \subseteq \delta(S_i)$.
\end{itemize}

Suppose for contradiction that for some distinct $i, j, k\in [3]$, there exists a path $P$ in $\mathcal{G}''$ from $s_i$ to some $v \in V_{jk}$. Since $\delta(Q)$ is a non-opposite cut,
the node $v$ is not in $S_i$. Also, since $v\in S_i'$ and we have seen above that $S_i'$ is disjoint from $S_j'$ and $S_k'$, it follows that $v\not\in S_j'\cup S_k'\supseteq S_j\cup S_k$. Hence, $v\not\in S_1\cup S_2\cup S_3$, and therefore $v\in B_r$ for some $r\in [4]$.

Let $u$ be the last node of $S_i$ on the path $P$. By property ($\star$), the end segment of $P$ starting at the node after $u$ is entirely in $\mathcal{G}[B_r]\setminus E_r'$. Since $\mathcal{G}''[S_i]$ is connected, we can replace the $s_i-u$ part of $P$ by a path in $\mathcal{G}''[S_i]$, and obtain an $s_i-v$ path in $\mathcal{G}''$ that uses only edges in $\G[S_i]\cup (\G[B_r]\setminus E_r')$ and a single edge in $E_r\subseteq\delta(S_i) \cap \delta(B_r)$. Hence, this is also a path in $E \setminus ((\delta(Q)\setminus E_r)\cup E_r')$. But we have already seen in Claim \ref{claim:changing_Q} that $(\delta(Q)\setminus E_r)\cup E_r'$ is a non-opposite cut-set, so $v \notin V_{jk}$, a contradiction.
\end{proof}

To show that $Cost_J(\delta(Q)) \geq Cost_J(F) + \ctwo\alpha$, we first observe that $E_i\subset \delta(B_i) \subset \delta(Q)$ for $i \in [4]$ and $E_i$'s are mutually disjoint since $\delta(B_i)$'s are mutually disjoint by Proposition \ref{prop:B_i_properties} \ref{item:Bprop4}. Therefore,
\begin{align*}
Cost_J(F) & \leq Cost_J(\delta(Q)\setminus (\cup_{i=1}^4 E_i)) + Cost_J(\cup_{i=1}^4 E_i')\\
& = Cost_J(\delta(Q)) - \sum_{i=1}^4 (Cost_J(E_i) - Cost_J(E_i')) \\
& \leq Cost_J(\delta(Q)) - \sum_{i=1}^4 \ctwo\beta_i \quad \quad \text{(By Claim \ref{claim:changing_Q})}\\
&\leq Cost_J(\delta(Q)) - \ctwo \alpha \quad \quad \text{(By Proposition \ref{prop:B_i_properties} \ref{item:Bprop5})}.
\end{align*}

\end{proof}

\subsection{Proof of Corollary \ref{coro:red-island-cut-cost}}
We restate and prove Corollary \ref{coro:red-island-cut-cost} now.
\coroRedIslandCutCost*
\begin{proof}[Proof of Corollary \ref{coro:red-island-cut-cost}]
Let $A:=A_1\cup A_2\cup A_3$.
We will show that $Cost_J(\delta(A))$ is at least $\ctwo \sum_{i=1}^3|A_i|/n^2-3/(2n)$ and that there exists a non-opposite cut $Q'$ satisfying $\delta(Q')=\delta(Q)\setminus\delta(A)$. 
By Lemma \ref{lemma:non_opposite_cut_cost},
$Cost_J(\delta(Q'))\geq 1.2-1/n$ and the corollary follows.

%
%
%

We first show a lower bound on the total cost of the edges in $\delta(A)$.

\begin{claim} \label{cl:S}
$Cost_J(\delta(A))\geq \ctwo \sum_{i=1}^3|A_i|/n^2-\frac{3}{2n}$.
\end{claim}
\begin{proof}

We will consider a specific non-opposite non-fragmenting cut to give a lower bound on the cost of $\delta(A)$ on $J$. Let $Q_0$ be defined as follows (see Figure \ref{fig:tight}):
\begin{equation*}
Q_0(x):=
\begin{cases}
1 & \text{if $x_1\geq 1/2$,}\\
2 & \text{if $x_1<1/2$, $x_2\geq 1/2$,}\\
3 & \text{otherwise.}
\end{cases}
\end{equation*}

\begin{figure}[ht]
\centering
\includegraphics[width=0.4\textwidth]{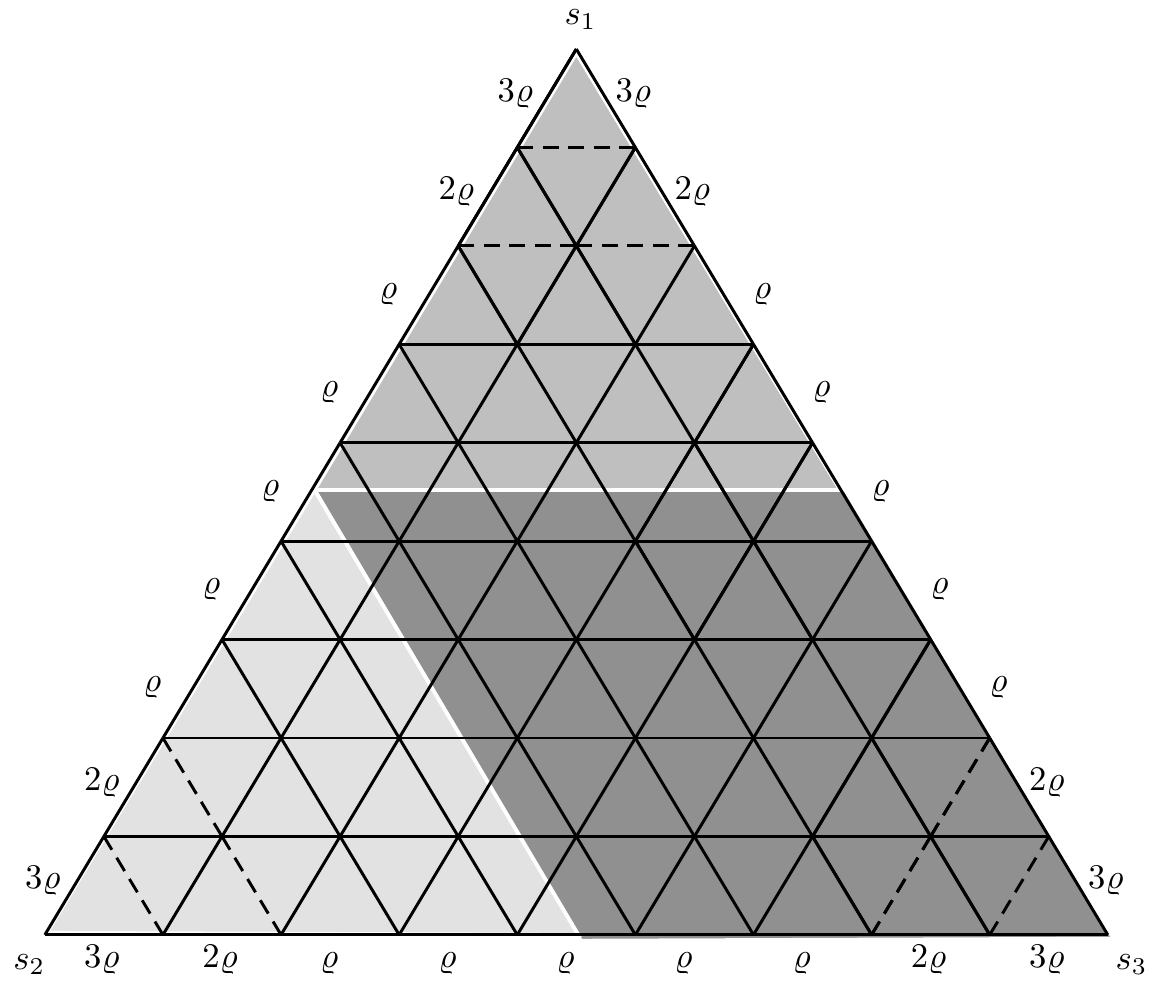}
\caption{The labeling $Q_0$.}
\label{fig:tight}
\end{figure}

Then, each edge in $\delta(Q_0)$ has a cost of $1.2/(2n)$ and the number of edges in $\delta(Q_0)$ is $2n+1$. Hence, $Cost_J(\delta(Q_0))\le 1.2+1/(2n)$. 
Moreover, $Q_0$ is a non-opposite non-fragmenting cut.  We now combine $\delta(A)$ and $\delta(Q_0)$ into a single cut by defining
\begin{equation*}
Q'_0(x):=
\begin{cases}
Q_0(x) & \text{if $x\not\in A$,}\\
4 & \text{otherwise.}
\end{cases}
\end{equation*}

We observe that $Q'_0$ is a non-opposite cut as it is obtained from a non-opposite cut by relabeling a subset of nodes that lie in the strict interior of $\closure(R_i)$ as $4$. As $A_i\neq\emptyset$ implies $\delta(A_i)\cap \Gamma_i=\emptyset$, we have that $\delta(Q'_0)$ intersects each side of the triangle the same number of times as $\delta(Q_0)$. That is, $Q'_0$ is also a non-fragmenting cut.
Therefore, we can apply Lemma~\ref{lemma:3-way-instance-properties} for $Q_0'$. The number of nodes labeled by $Q_0'$ as $4$ is exactly $|A|$. Hence,
\[
Cost_J(\delta(Q_0'))\ge 1.2-\errorterm+\ctwo\frac{|A|}{n^2}.
\]
By $Q_0(v)=i$ for each $v\in\closure(R_i)$ and by $\delta(A_i)\cap \Gamma_i=\emptyset$ for $i\in[3]$, we have $\delta(Q'_0)=\delta(A)\cup\delta(Q_0)$, implying
\[
Cost_J(\delta(A))+Cost_J(\delta(Q_0))\ge Cost_J(\delta(A)\cup\delta(Q_0))=Cost_J(\delta(Q_0'))\ge 1.2-\errorterm+\ctwo\frac{|A|}{n^2}.
\]
We recall that $Cost_J(\delta(Q_0))\leq 1.2+1/(2n)$. Hence, $Cost_J(\delta(A))\ge \ctwo|A|/n^2-3/(2n)$.
\end{proof}

Let $K\subseteq [3]$ denote the set of indices $i$ for which $\delta(Q)\cap \Gamma_i=\emptyset$ and let $Q'$ be a labeling obtained from $Q$ by setting
\[
Q'(v):=
\begin{cases}
i & \text{ if $v\in \closure(R_i)$ for some $i\in K$}, \\
Q(v) & \text{ otherwise}.
\end{cases}
\]

\begin{claim} \label{cl:relabel}
$Q'$ is a non-opposite cut with $Cost_J(\delta(Q'))\leq Cost_J(\delta(Q))-Cost_J(\delta(A))$.
\end{claim}
\begin{proof}
The cut $Q$ is non-opposite and $Q'(v)\in \supp(v)$ for each relabeled node $v$, hence $Q'$ is also a non-opposite cut. For any index $i\in K$, $\delta(Q)\cap \Gamma_i=\emptyset$ implies $Q(v)=i$ for $v\in R_i$. Thus, we have $\delta(Q')\subseteq \delta(Q)\setminus \delta(A)$ and the claim follows.
\end{proof}

As $Q'$ is a non-opposite cut, Lemma \ref{lemma:non_opposite_cut_cost} implies $Cost_J(\delta(Q'))\geq 1.2-1/n$. By Claim~\ref{cl:S}, $Cost_J(\delta(A))\geq \ctwo \sum_{i=1}^3|A_i|/n^2-3/(2n)$. These together with Claim~\ref{cl:relabel} imply that $Cost_J(\delta(Q))\geq 1.2-5/(2n)+\ctwo \sum_{i=1}^3|A_i|/n^2$, finishing the proof of the corollary.
\end{proof}

\section{Limitations of our instances}\label{sec:limitations}
In this section, we will show two results. In Section \ref{sec:insufficiency-of-3-instances}, we will show that instances $I_1$, $I_2$ and $I_4$ are insufficient to obtain a gap better than $1.2$. This result also motivates our choice of instance $I_3$. 
In Section \ref{sec:four-instances-limitations}, we will show that instances $I_1$, $I_2$, $I_3$ and $I_4$ are insufficient to obtain a gap better than $1.20067$, thus exhibiting a limitation of our choice of instance $I_3$. 
\subsection{Insufficiency of instances $I_1$, $I_2$ and $I_4$ to beat $1.2$}\label{sec:insufficiency-of-3-instances}
In this section, we will show that convex combinations of instances $I_1$, $I_2$ and $I_4$ are insufficient to obtain an instance which has gap larger than $1.2$ against non-opposite cuts. For this, we will exhibit two non-opposite cuts $P$ and $P'$ such that at least one of them will have cost at most $1.2$ in every convex combination of instances $I_1$, $I_2$ and $I_4$. 
\begin{enumerate}
\item Consider the cut $Q_0$ in $\Delta_{3,n}$ defined in the proof of Claim \ref{cl:S}. 
Extend it to a cut $P$ in $\Delta_{4,n}$ as follows:
\begin{align*}
P(x):=
\begin{cases}
Q_0(x) & \text{ if $x_4=0$},\\
4 &\text{ if $x_4>0$}.
\end{cases}
\end{align*}
Then, $P$ is a non-opposite cut with the cost of $P$ on $I_1$, $I_2$ and $I_4$ being $1.2+O(1/n)$, $1$ and $1.5+O(1/n^2)$ respectively. Hence, the cost of $P$ on the convex combination $\lambda_1 I_1+\lambda_2 I_2 + \lambda_4 I_4$ is at most $1.2\lambda_1 + \lambda_2 + 1.5\lambda_4+O(1/n)$. 
\item Consider the cut $P'$ in $\Delta_{4,n}$ defined as follows:
\begin{align*}
P'(x):=
\begin{cases}
i &\text{ if $x=e^i$},\\
5 &\text{ otherwise}.
\end{cases}
\end{align*}
Then, $P'$ is a non-opposite cut with the costs of $P'$ on $I_1$, $I_2$ and $I_4$ being $1.2$, $2$ and $O(1/n^2)$ respectively. Hence, the cost of $P'$ on the convex combination $\lambda_1 I_1+\lambda_2 I_2 + \lambda_4 I_4$ is at most $1.2\lambda_1 + 2\lambda_2+O(1/n^2)$. 
\end{enumerate}
Consequently, for every convex combination defined by $\lambda_1, \lambda_2, \lambda_4$, there exists a non-opposite cut whose cost on the convex combination instance $\lambda_1 I_1 + \lambda_2 I_2 + \lambda_4 I_4$ is at most 
\[
\min \left\{1.2\lambda_1 + \lambda_2 + 1.5\lambda_4, 1.2\lambda_1 + 2\lambda_2\right\}+O(1/n).
\]
The following claim shows that the above expression is at most $1.2+O(1/n^2)$ for every convex combination. 
\begin{claim}
For every $\lambda_1, \lambda_2, \lambda_4\ge 0$ with $\lambda_1+\lambda_2+\lambda_4=1$, we have 
\[
\min \left\{1.2\lambda_1 + \lambda_2 + 1.5\lambda_4, 1.2\lambda_1 + 2\lambda_2\right\} \le 1.2.
\]
\end{claim}
\begin{proof}
Say not. Then, both expressions are greater than $1.2$. 
\begin{enumerate}
\item We have $1.2\lambda_1 + \lambda_2 + 1.5\lambda_4>1.2$ which implies that $\lambda_2 + 1.5\lambda_4>1.2(1-\lambda_1)=1.2\lambda_2 + 1.2\lambda_4$. Hence, $3\lambda_4>2\lambda_2$. 
\item We have $1.2\lambda_1 + 2\lambda_2>1.2$, which implies that $2\lambda_2>1.2(1-\lambda_1)=1.2\lambda_2 + 1.2\lambda_4$. Hence, $2\lambda_2>3\lambda_4$, a contradiction. 
\end{enumerate}
\end{proof}

We emphasize that instance $I_3$ is constructed specifically to boost the cost against the non-opposite cut $P'$. 
\subsection{Limitation of instances $I_1,I_2,I_3,$ and $I_4$}\label{sec:four-instances-limitations}

In this section, we will show that convex combinations of instances $I_1,I_2,I_3,$ and $I_4$ are insufficient to obtain an instance which has gap larger than $1.20067$ against non-opposite cuts. 

\begin{proof}[Proof of Theorem \ref{thm:limitation}]
To show this, we exhibit three non-opposite cuts $P_1,P_2$ and $P_3$ such that at least one of them will have cost at most $1.20067 + O(1/n)$ in every convex combination of instances $I_1,I_2,I_3$, and $I_4$.  
\begin{enumerate}
\item $P_1$ is same as the non-opposite cut $P$ defined in Section \ref{sec:insufficiency-of-3-instances}. The cost of $P_1$ on instances $I_1, I_2,I_3$, and $I_4$ are $1.2+O(1/n), 1, 0,$ and $1.5+O(1/n^2)$ respectively. Consequently, the cost of $P_1$ on the convex combination $\lambda_1I_1 + \lambda_2 I_2 + \lambda_3 I_3 + \lambda_4 I_4$ is at most $1.2\lambda_1 + \lambda_2 + 1.5 \lambda_4 + O(1/n)$. 
\item $P_2$ is same as the non-opposite cut $P'$ defined in Section \ref{sec:insufficiency-of-3-instances}. The cost of $P_2$ on $I_1,I_2,I_3$, and $I_4$ are $1.2, 2, 6/9c$, and $O(1/n^2)$ respectively. Consequently, the cost of $P_2$ on the convex combination $\lambda_1 I_1 + \lambda_2 I_2 + \lambda_3 I_3 + \lambda_4 I_4$ is at most $1.2 \lambda_1 + 2\lambda_2 + \frac{6}{9c} \lambda_3 + O(1/n^2)$. 
\item Let $P_3$ be defined as follows:
\begin{align*}
P_3(x):=
\begin{cases}
4 & \text{ if $x = e^4$,}\\
i &\text{ if $x_4 = 0, x_i\geq 1-c, i \in\{1,2,3\}$},\\
5 &\text{ otherwise}.
\end{cases}
\end{align*}
Then, $P_3$ is a non-opposite cut with the cost of $P_3$ on $I_2,I_3$, and $I_4$ being $2,0,$ and $9c^2/2 + O(1/n)$ respectively. Moreover, if $c< 1/9$, then cost of $P_3$ on $I_1$ is $1.2$. Hence, if $c<1/9$, then the cost of $P_3$ on the convex combination $\lambda_1 I_1 + \lambda_2 I_2 + \lambda_3 I_3 + \lambda_4 I_4$ is at most $1.2 \lambda_1 + 2 \lambda_2 + (9c^2/2) \lambda_4 + O(1/n)$. 
\end{enumerate}
Consequently, there exists a non-opposite cut on instance $\lambda_1 I_1 + \lambda_2 I_2 + \lambda_3 I_3 + \lambda_4 I_4$ whose cost is 
\begin{align*}
&\text{at most } \min \left\{1.2 \lambda_1 + \lambda_2 + 1.5 \lambda_4, 1.2 \lambda_1 + 2\lambda_2 + \frac{6}{9c} \lambda_3\right\} + O(1/n)\ \text{if $c \ge 1/9$ and} \\ 
&\text{at most } \min \left\{1.2 \lambda_1 + \lambda_2 + 1.5 \lambda_4, 1.2 \lambda_1 + 2\lambda_2 + \frac{6}{9c} \lambda_3, 1.2 \lambda_1 + 2\lambda_2 + \frac{9c^2}{2} \lambda_4\right\} + O(1/n) \ \text{if $c < 1/9$}.
\end{align*}
The following claim shows that the above terms are at most $1.20067 + O(1/n)$ for every convex combination, thus completing the proof of the theorem.
\end{proof}

\begin{claim}
For every $\lambda_1,\lambda_2,\lambda_3,\lambda_4 \geq 0$ with $\lambda_1 + \lambda_2 + \lambda_3 + \lambda_4 = 1$, we have 
\begin{enumerate}
\item $\min \left\{1.2 \lambda_1 + \lambda_2 + 1.5 \lambda_4, 1.2 \lambda_1 + 2\lambda_2 + \frac{6}{9c} \lambda_3\right\} \leq 1.2$ if $c\geq 1/9$, and 
\item $\min \left\{1.2 \lambda_1 + \lambda_2 + 1.5 \lambda_4, 1.2 \lambda_1 + 2\lambda_2 + \frac{6}{9c} \lambda_3, 1.2 \lambda_1 + 2\lambda_2 + \frac{9c^2}{2} \lambda_4\right\}\leq 1.20067$ if $c<1/9$.
\end{enumerate}
\end{claim}
\begin{proof}
Minimum of a set of values is at most the convex combination of the values. 
\begin{enumerate}
\item Suppose $c \geq 1/9$. Then, 
\begin{align*}
&\min \left\{1.2 \lambda_1 + \lambda_2 + 1.5 \lambda_4, 1.2 \lambda_1 + 2\lambda_2 + \frac{6}{9c} \lambda_3\right\}\\
&\quad \le 0.8 (1.2 \lambda_1 + \lambda_2 +1.5 \lambda_4) + 0.2 \left(1.2 \lambda_1 + 2\lambda_2 + \frac{6}{9c} \lambda_3\right)\\
&\quad = 1.2 \lambda_1 + 1.2 \lambda_2 + 1.2 \lambda_4 + \frac{1.2}{9c} \lambda_3 \\
&\quad \leq 1.2 \lambda_1 + 1.2 \lambda _2 + 1.2 \lambda_4 + 1.2 \lambda_3 & \qquad (\text{since } c\geq 1/9)\\
&\quad = 1.2. 
\end{align*}
The last equality above is because $\lambda_1 + \lambda_2 + \lambda_3 + \lambda_4 = 1$.

\item Suppose $c<1/9$. 
Let 
\[
\beta := \max\left\{\frac{3-9c^2/2}{2.5 - 9c^2/2 + 27c^3/4}:\ 0 \leq c <\frac{1}{9}\right\}.
\]
Then, it is straightforward to verify that $1.2\le \beta \leq 1.20067$.  For $c<1/9$, it follows that $\beta (1-3c/2) \geq 1$. Hence, the multipliers $(2-\beta), 3c\beta/2,\beta-3c\beta/2-1$ are non-negative and sum to one. Therefore, 
\begin{align}
&\min \left\{1.2 \lambda_1 + \lambda_2 + 1.5 \lambda_4, 1.2 \lambda_1 + 2\lambda_2 + \frac{6}{9c} \lambda_3, 1.2 \lambda_1 + 2\lambda_2 + \frac{9c^2}{2} \lambda_4\right\} \notag\\
&\quad \le (2-\beta)(1.2 \lambda_1 + \lambda_2 + 1.5 \lambda_4) + \left(\frac{3c\beta}{2}\right) \left(1.2 \lambda_1 + 2\lambda_2 + \frac{6}{9c} \lambda_3\right) \notag\\
&\quad \qquad \qquad+ \left(\beta-\frac{3c\beta}{2}- 1\right) \left(1.2 \lambda_1 + 2\lambda_2 + \frac{9c^2}{2} \lambda_4\right)\notag\\
&\quad =1.2 \lambda_1 + \beta \lambda_2 + \beta \lambda_3 + \left(3 - \frac{9c^2}{2} - \beta \left(1.5-\frac{9c^2}{2} +\frac{27c^3}{4}\right) \right)\lambda_4\notag\\
&\quad \leq 1.2 \lambda_1 + \beta\lambda_2 + \beta \lambda_3 + \beta \lambda_4\label{eq:beta-defn}\\
&\quad \leq \beta \label{eq:beta-LB}\\
&\quad \leq 1.20067. \notag
\end{align}
Inequality (\ref{eq:beta-defn}) follows from the definition of $\beta$ and inequality (\ref{eq:beta-LB}) follows from the fact that $1.2 \leq \beta$ and $\lambda_1 + \lambda_2 + \lambda_3 + \lambda_4 = 1$.
\end{enumerate}
\end{proof}

\section*{Acknowledgements}

Tam\'{a}s was supported by the Hungarian National Research, Development and Innovation Office -- NKFIH grant K120254. Krist\'of was supported by the \'UNKP-18-4 New National Excellence Program of the Ministry of Human Capacities.

\bibliographystyle{alpha}
\bibliography{references}

\newcommand{\etalchar}[1]{$^{#1}$}
\begin{thebibliography}{MNRS08}

\bibitem[AMM17]{AMM17}
H.~Angelidakis, Y.~Makarychev, and P.~Manurangsi.
\newblock An improved integrality gap for the
  {C}{\u{a}}linescu-{K}arloff-{R}abani relaxation for multiway cut.
\newblock In {\em Integer Programming and Combinatorial Optimization}, IPCO
  '17, pages 39--50, 2017.

\bibitem[BNS13]{BNS13}
N.~Buchbinder, J.~Naor, and R.~Schwartz.
\newblock Simplex partitioning via exponential clocks and the multiway cut
  problem.
\newblock In {\em Proceedings of the Forty-fifth Annual ACM Symposium on Theory
  of Computing}, STOC '13, pages 535--544, 2013.

\bibitem[BSW17]{BSW17}
N.~Buchbinder, R.~Schwartz, and B.~Weizman.
\newblock Simplex transformations and the multiway cut problem.
\newblock In {\em Proceedings of the Twenty-Eighth Annual ACM-SIAM Symposium on
  Discrete Algorithms}, SODA '17, pages 2400--2410, 2017.

\bibitem[CCT06]{CCT06}
K.~Cheung, W.~Cunningham, and L.~Tang.
\newblock Optimal 3-terminal cuts and linear programming.
\newblock {\em Mathematical Programming}, 106(1):1--23, Mar 2006.

\bibitem[CKR00]{CKR00}
G.~C{\u{a}}linescu, H.~Karloff, and Y.~Rabani.
\newblock An improved approximation algorithm for multiway cut.
\newblock {\em Journal of Computer and System Sciences}, 60(3):564 -- 574,
  2000.

\bibitem[DJP{\etalchar{+}}94]{DJPSY94}
E.~Dahlhaus, D.~Johnson, C.~Papadimitriou, P.~Seymour, and M.~Yannakakis.
\newblock The complexity of multiterminal cuts.
\newblock {\em SIAM Journal on Computing}, 23(4):864--894, 1994.

\bibitem[FK00]{FK00}
A.~Freund and H.~Karloff.
\newblock {A lower bound of 8/(7+1/(k-1)) on the integrality ratio of the
  {C}{\u{a}}linescu-{K}arloff-{R}abani relaxation for multiway cut}.
\newblock {\em Information Processing Letters}, 75(1):43 -- 50, 2000.

\bibitem[KKS{\etalchar{+}}04]{KKSTY04}
D.~Karger, P.~Klein, C.~Stein, M.~Thorup, and N.~Young.
\newblock Rounding algorithms for a geometric embedding of minimum multiway
  cut.
\newblock {\em Mathematics of Operations Research}, 29(3):436--461, 2004.

\bibitem[MNRS08]{MNRS08}
R.~Manokaran, J.~Naor, P.~Raghavendra, and R.~Schwartz.
\newblock {SDP} gaps and {UGC} hardness for multiway cut, 0-extension, and
  metric labeling.
\newblock In {\em Proceedings of the Fortieth Annual ACM Symposium on Theory of
  Computing}, STOC '08, pages 11--20, 2008.

\bibitem[MV15]{MV15}
M.~Mirzakhani and J.~Vondr\'{a}k.
\newblock Sperner's colorings, hypergraph labeling problems and fair division.
\newblock In {\em Proceedings of the Twenty-sixth Annual ACM-SIAM Symposium on
  Discrete Algorithms}, SODA '15, pages 873--886, 2015.

\bibitem[SV14]{SV14}
A.~Sharma and J.~Vondr\'{a}k.
\newblock Multiway cut, pairwise realizable distributions, and descending
  thresholds.
\newblock In {\em Proceedings of the Forty-sixth Annual ACM Symposium on Theory
  of Computing}, STOC '14, pages 724--733, 2014.

\end{thebibliography}

\end{document}